\documentclass[a4paper,10pt]{article}
\usepackage{enumerate}
\usepackage[nointlimits]{amsmath}
\usepackage{amsfonts}
\usepackage{amssymb}
\usepackage{amsmath}
\usepackage{amsthm}
\usepackage{latexsym}
\usepackage{graphicx}
\usepackage{layout}
\usepackage[english]{babel}
\usepackage{fancyhdr}

\usepackage{lscape} 

\usepackage{amscd}

\usepackage{color}

\newtheorem{lemma}     {Lemma}[section]

\newtheorem{teorema1}   [lemma]{Theorem}
\newtheorem{prop}      [lemma]{Proposition}
\newtheorem{coro}    [lemma]{Corollary}
\newtheorem{cong1}      [lemma]{Conjecture}
\newtheorem{remark1}    [lemma]{Remark}

\numberwithin{equation}{section}

\renewcommand{\(}{\left(}        \renewcommand{\)}{\right)}
\renewcommand{\[}{\left[}        \renewcommand{\]}{\right]}

\newcommand{\R}{\mathbb R}
 \renewcommand{\d}{\delta}
 
 \renewcommand{\a}{\alpha}
 \renewcommand{\b}{\beta}
 \newcommand{\g}{\gamma}

\newcommand{\dis}{\displaystyle}

\newcommand{\mmmintone}[1]{{\dis{\int\kern -.43cm
-}}_{\kern-.21cm\substack{#1}}\;\;}
\newcommand{\mmmintwo}[2]{{\dis{\int\kern -.43cm
-}}_{\kern-.21cm\substack{#1}}^{\substack{#2}}\;\;}
\newcommand{\submint}{{\scriptstyle{\int\kern -.66em -}}}
\newcommand{\submintone}[1]{{\scriptstyle{\int\kern -.66em
-}}_{\scriptscriptstyle{\kern-.21em\substack{#1}}}}
\newcommand{\fracmint}{{\textstyle{\int\kern -.88em -}}}
\newcommand{\fracmintone}[1]{{\textstyle{\int\kern -.88em
-}}_{\scriptscriptstyle{\kern-.21em\substack{#1}}}\;}

\newcommand{\eps}{\epsilon}

\newcommand{\ga}{\gamma}

\newcommand{\nada}[1]{}
\newcommand{\bea}{\begin{eqnarray}}
\newcommand{\ea}{\end{eqnarray}}

\newcommand{\be}{\begin{equation}}
\newcommand{\ee}{\end{equation}}
\newcommand{\beq}{\begin{eqnarray}}
\newcommand{\eeq}{\end{eqnarray}}

\begin{document}
\title{Duality for stochastic models of transport}
\author{
Gioia Carinci$^{\textup{{\tiny(a)}}}$,
Cristian Giardin\`a$^{\textup{{\tiny(a)}}}$,
Claudio Giberti$^{\textup{{\tiny(b)}}}$,
Frank Redig$^{\textup{{\tiny(c)}}}$.
\\
\\
{\small $^{\textup{(a)}}$
University of Modena and Reggio Emilia},
{\small via G. Campi 213/b, 41125 Modena, Italy}
\\
{\small $^{\textup{(b)}}$
University of Modena and Reggio Emilia},
{\small Via Amendola 2, Pad. Morselli, 42122 Reggio Emilia, Italy}
\\
{\small $^{\textup{(c)}}$
University of Delft}, {\small Mekelweg 4 2628 CD Delft , The Netherlands}
\\
}
\maketitle

\pagenumbering{arabic}

\begin{abstract}
We study three classes of continuous time Markov processes (inclusion process,
exclusion process, independent walkers) and a family of
interacting diffusions (Brownian energy process).
For each model we define a boundary
driven process which is obtained by placing the system in contact
with proper reservoirs, working at different {particle} densities or
different temperatures. We show that {all the models}
are exactly solvable by duality, using a dual process with absorbing
boundaries. {The solution does also apply to the so-called thermalization
limit in which particles or energy is instantaneously redistributed among
sites.}

The results shows that duality is a versatile tool for analyzing stochastic
models of transport, while the analysis in the literature has been
so far limited to particular {instances}.
Long-range correlations naturally emerge as a result of the interaction
of dual particles at the microscopic level and the {explicit computations of covariances}
match, in the scaling limit, { the predictions of the macroscopic fluctuation theory}.
\end{abstract}

\vspace{1.cm}
\section{Introduction}
Interacting particle systems are classical models to study non-equilibrium statistical mechanics.
The standard setting is the one in which a system is placed in contact with reservoirs working
at different parameters that create a stationary state characterized by a non-zero averaged
current.
The prototypical example are the Symmetric Exclusion process with at most one
particle per site connected to birth and death process at the boundaries \cite{L, D} and the KMP process \cite{KMP} connected to reservoirs which impose
at the boundaries Boltzmann-Gibbs distribution with different temperatures.
The Symmetric exclusion process is a model for transport of a discrete quantity,
whereas the KMP process models transport of a continuous quantity.

Problems that are very hard for classical Hamiltonian systems -- for instance deriving Fourier law starting from the microscopic evolution --
can be successfully approached using stochastic models. Furthermore  stochastic models
of transport have been used to prove new theorems in non-equilibrium statistical mechanics, such as the fluctuation theorem \cite{GC,K}, to introduce new principles,
such as the additivity principle \cite{BD}, to construct new schemes, such
as the macroscopic fluctuation theory that describes the density and current
large deviations for diffusive systems \cite{BDGJL1,DLS}, to test new algorithms,
such as cloning algorithms to simulate rare events \cite{GKLT}.
Recently, the connection between deterministic Hamiltonian systems and stochastic models is
emerging either by considering evolutions in which they are coupled \cite{BO} or by
considering slow/fast variables \cite{DL} and thermodynamic formalismo \cite{LAW}.

An important tool in the study of interacting stochastic systems is duality
\cite{S, L}.
Duality provides the connection between a process and a simpler dual process.
This technique has been applied in different contexts, including interacting
particles systems, interacting diffusions, queueing theory and
mathematical population genetics. For a recent review on duality,
which also include many references, see \cite{JK}. For recent applications
of duality in the context of asymmetric processes and KPZ universality see
\cite{corwin}.

In the context of interacting particle systems or interacting diffusion processes
modeling non-equilibrium
systems, the main simplification coming from duality lies in the fact that
for an appropriate choice of the modeling of the boundary
reservoirs, a dual process exists where the reservoirs are replaced by absorbing boundaries.
This was originally found for
the boundary driven Symmetric Exclusion process with at most one particle
per site \cite{Spo} and for the KMP model \cite{KMP}.
As a consequence, the $n$-point correlation
functions in the non-equilibrium steady state can be obtained from absorption probabilities
of $n$ dual particles. In particular, the stationary density or temperature profile
can be easily obtained from a single dual walker.
Other simplifications due to duality include ``from continuous to discrete'', i.e., connecting continuous systems with discrete
particle systems and ``from many to few'', i.e., correlation functions in a systems of possibly infinitely
many particles reduce to as many dual particles as the degree of the correlation function.

In this paper we introduce and study a large class of boundary driven processes which can be dealt with via this technique of
duality. We treat processes  with interactions of ``inclusion'' (attractive) and ``exclusion'' (repulsive) type.

The particle systems
range from the Symmetric Inclusion Processes (SIP) with  Negative-Binomial
product stationary measures at equilibrium, to the Symmetric Exclusion Processes (SEP) having a Binomial product measures as
equilibrium state, via Independent Random Walkers (IRW) with a product Poisson stationary
measures. The interacting diffusions corresponding to the SIP are given by the so-called Brownian
Energy processes (BEP), having product of Gamma distributions as equilibrium.

We also study ``thermalized versions'' of these processes. For the diffusion models
thermalization leads to ``energy redistribution models'' of which the famous KMP model is a particular instance.
For particle systems thermalization leads to ``occupation redistribution models'' where
in one event associated to a nearest neighbour edge, occupations of particles are reshuffled according to a specific redistribution measure. The dual KMP model is a particular instance of these
thermalized particle systems. Most of these thermalized models are new, as well as their boundary driven versions.
A non-trivial stationary state is found for these boundary driven thermalized models even considering
only one site, since the reservoirs are not additive.

Some of the processes we discuss here have already been introduced before: we have chosen to
include all of them, including independent random walkers, in order to provide a (up to know and
to our knowledge) complete and self-contained overview of the interacting non-equilibrium systems
that can be treated with duality.
The main message of this paper is thus an extension of duality and its consequences into the boundary driven non-equilibrium setting
for all the models discovered and studied in \cite{GKR, GKRV, GRV}.

\section{Models definition}\label{systems}
In this section we introduce
our models.
In the most complete setting,
they are constituted by a bulk  which is kept in a non equilibrium state
by the contact with particles or energy reservoirs.
In particular, we consider one-dimensional systems
on a finite lattice $\{1,\ldots,L\}$, whose boundaries
(i.e. sites $1$ and $L$)
interact with the reservoirs. When needed, the reservoirs themselves will be represented by two extra sites, namely sites $0$ and $L+1$.\\
Accordingly, the generators of the random processes associated with our models can be generically expressed as the sum of three terms
\be
{\cal L}={\cal L}_a+{\cal L}_0+{\cal L}_b\;,
\ee
where ${\cal L}_0$ represents the generator of the dynamics in the bulk, while ${\cal L}_a$ and ${\cal L}_b$ represent the generators of the reservoirs.\\
We will consider four models: three classes of interacting particle systems, characterized by the different interactions  between the particles, and one family of interacting diffusions introduced to model heat conduction
\cite{GKR, GKRV, GRV}. The models are:
\begin{enumerate}
\item the Symmetric Inclusion Process (SIP), with {\em attractive} interaction between neighbouring particles;
\item the Symmetric Exclusion Process  (SEP), with {\em repulsive} interaction between neighbouring particles;
\item the Independent Random Walkers (IRW), without interactions among particles;
\item the Brownian Energy Process (BEP).
\end{enumerate}
In the first three cases the dynamic variable is a vector that specifies the number of particle on each site:  $\eta=(\eta_1, ...,\eta_L)\in \Omega$; here  $\Omega$,
the state space, depends on the model  and will be defined ahead.
In the case of the BEP the dynamic variable is a vector $z$ representing the energies on each site of the lattice: $z=(z_1,\ldots,z_L)\in \Omega \equiv  {\mathbb R}_+^L$.

\subsection{Interacting particle systems}
The generators of the reservoirs for SIP, SEP and IRW have the following general form:
\bea
{\cal L}_{a} f(\eta)&=&b(\eta_1)[f(\eta^{0,1})-f(\eta)]+d(\eta_1)[f(\eta^{1,0})-f(\eta)],\label{terma}\\
{\cal L}_{b} f(\eta)&=&b(\eta_L)[f(\eta^{L+1,L})-f(\eta)]+d(\eta_L)[f(\eta^{L,L+1})-f(\eta)]\label{termb}.
\ea
Here
$\eta^{i,i+1}$  denotes the configuration obtained
from $\eta$ by moving a particle from site $i$ to site $i+1$,
i.e. $\eta^{i,i+1} = (\eta_1,\ldots,\eta_i-1,\eta_{i+1}+1,\ldots,\eta_L)$.
According to (\ref{terma}) and (\ref{termb})  particles are injected into the system through the boundaries
with rate $b(n)$ with $n\in\mathbb{N}_0$,  and removed from the same sites with rate $d(n)$.  While $b(n)$ is
model-dependent, the annihilation rate is not,  being in any case proportional to the number of
particles at the boundary site.\\We introduce now our models by defining the actions of the generators ${\cal L}$ on the functions $f:\Omega \rightarrow \R$. 


\vspace{0.5cm}
{\bf Inclusion walkers SIP($2k$).} The inclusion process (without boundaries) is introduced
first in \cite{GKR}, and also studied further in \cite{GRV}.

In the SIP($2k$), see Figure \ref{figSIP},  each site can accomodate
an arbitrary number of particles, thus $\Omega={\mathbb N}^{L}_0$.
In the bulk each particle may jump to its left or right  neighbouring site with rates proportional to the number
of particles in the departure site and to the number of particles in the arrival site. In each  boundary site particles are created with a rate proportional to $2k$ plus the number
of particles sitting in that site; $k\in \mathbb R_+$  labels the class of models.
The generator is
\begin{eqnarray}\label{INC}
{\cal L}^{SIP}f(\eta)&=&{\cal L}_a^{SIP}f(\eta)+{\cal L}_0^{SIP}f(\eta)+{\cal L}_b^{SIP}f(\eta)\\
&=&\a (2k+\eta_1)\[f(\eta^{0,1})-f(\eta)\]+\ga \eta_1 \[f(\eta^{1,0})-f(\eta)\]\nonumber\\
&+& \sum_{i=1}^{L-1}\eta_i(2k+\eta_{i+1})\[f(\eta^{i,i+1})-f(\eta)\]+ \eta_{i+1}(2k+\eta_i)\[f(\eta^{i+1,i})-f(\eta)\]\nonumber\\
&+&\d (2k+\eta_L)\[f(\eta^{L+1,L})-f(\eta)\]+\b \eta_L \[f(\eta^{L,L+1})-f(\eta)\].\nonumber
\end{eqnarray}
The positive numbers $\a$ and $\ga$ (resp. $\d$ and $\b$) tune the creation and annihilation rates of  the left  (resp. right) reservoirs.
\begin{figure}[h]
\centering
\includegraphics[width=13cm]{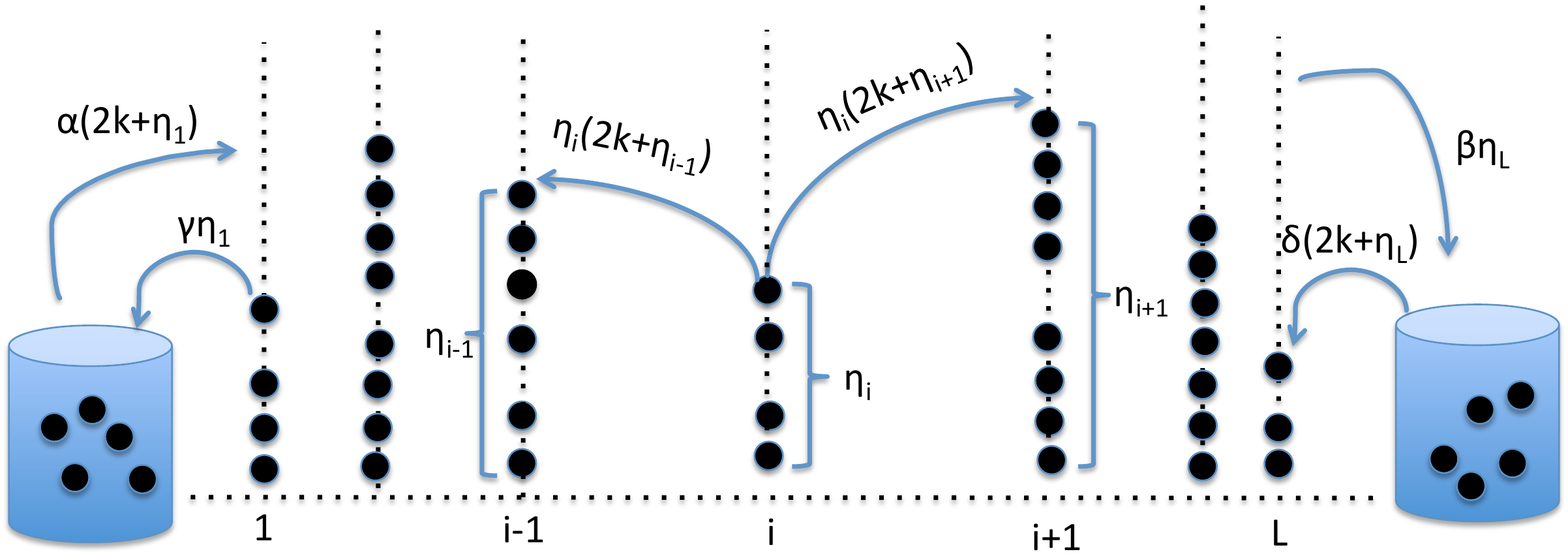}
\vspace{-4cm}
\caption{\label{figSIP} Schematic description of the Symmetric Inclusion Process SIP($2k$). The arrows represent the possible transitions and the corresponding rates, while the two cylinders represent the boundary reserovoirs. Each site can accomodate an arbitrary numer of particles. }
\end{figure}

\vspace{0.5cm}
{\bf Exclusion walkers SEP($2j$).}
For $j=1/2$ the boundary driven simple exclusion process has been studied using duality
in \cite{Spo}. The model for arbitrary $j$ has been introduced and studied in \cite{SS}. 
From the mathematical point of view
a related model which also exhibits product measures, but which does not have the self-duality property
is studied (without boundary reservoirs)  in \cite{Kies}.

In the SEP($2j$) the maximum occupation number at each site is $2j\in \mathbb{N}$, thus $\Omega=\{0,1,\ldots, 2j\}^{L}$. In the bulk
particles jump independently to nearest neighbouring lattices sites at rate
proportional to the number of particles in the departure site times
the number of holes in the arrival site. The reservoirs inject particles in the systems with a rate proportional to the holes in the boundary sites, see Figure \ref{figSEP}.
The generator is
\begin{eqnarray}\label{EXC}
{\cal L}^{SEP}f(\eta)&=&{\cal L}_a^{SEP}f(\eta)+{\cal L}_0^{SEP}f(\eta)+{\cal L}_b^{SEP}f(\eta)\\
&=&\a (2j-\eta_1)\[f(\eta^{0,1})-f(\eta)\]+\ga \eta_1 \[f(\eta^{1,0})-f(\eta)\]\nonumber\\
&+& \sum_{i=1}^{L-1}\eta_i(2j-\eta_{i+1})\[f(\eta^{i,i+1})-f(\eta)\]+ \eta_{i+1}(2j-\eta_i)\[f(\eta^{i+1,i})-f(\eta)\]\nonumber\\
&+&\d (2j-\eta_L)\[f(\eta^{L+1,L})-f(\eta)\]+\b \eta_L \[f(\eta^{L,L+1})-f(\eta)\]\;.\nonumber
\end{eqnarray}
The parameters  $\a,\,\ga,\,\d,\, \b$ have the same meaning as in the SIP($2k$).
\begin{figure}[h]
\centering
\includegraphics[width=13cm]{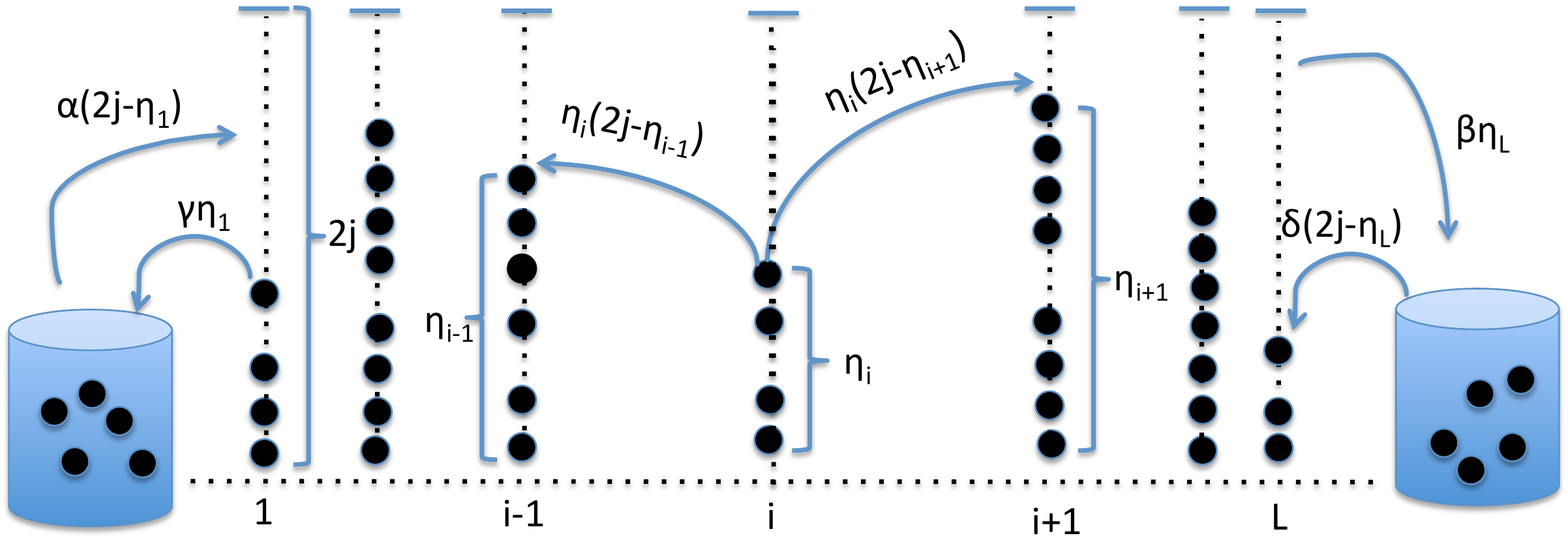}
\vspace{-4cm}
\caption{\label{figSEP} Schematic description of the Symmetric Exclusion Process SEP($2j$). The arrows represent the possible transitions and the corresponding rates, while the two cylinders represent the boundary reserovoirs. Each site can accomodate up to $2j$ particles. }
\end{figure}

\vspace{0.5cm}
{\bf Independent random walkers IRW.}
This well-known model is first considered in \cite{S}, and with boundaries is also well-known and studied
e.g. in \cite{Levine} (where also the more general boundary driven zero range process is studied).

In the IRW model each particle jumps independently to nearest neighbouring lattices sites at rate 1, and each site can accomodate an arbitrary number of particles, thus
$\Omega={\mathbb N}^{L}_0$.
Jumps occur with the same probability to the right and to the left, while particles are created at rates $\a$ and $\d$ irrespective of the numer of particles at the boundaries.
Therefore the system is described by the generator
\begin{eqnarray}\label{IND}
{\cal L}^{IRW}f(\eta)&=&{\cal L}_a^{IRW}f(\eta)+{\cal L}_0^{IRW}f(\eta)+{\cal L}_b^{IRW}f(\eta)\\
&=&\a \[f(\eta^{0,1})-f(\eta)\]+ \ga \eta_1 \[f(\eta^{1,0})-f(\eta)\]\nonumber \\
&+& \sum_{i=1}^{L-1}\eta_i\[f(\eta^{i,i+1})-f(\eta)\]+ \eta_{i+1}\[f(\eta^{i+1,i})-f(\eta)\]\nonumber\\
&+&\d\[f(\eta^{L+1,L})-f(\eta)\]+ \b \eta_L \[f(\eta^{L,L+1})-f(\eta)\]\;.\nonumber
\end{eqnarray}
The dynamics in the bulk can be further described by saying that if at site $i$ there are $\eta_i$ particles, one of the particle jumps
at rate $\eta_i$ either to the left or to the right. As in the previous cases, parameters $\g$ and $\b$ define
the annihilation processes.

\begin{remark1}
{The effect of the reservoirs is to impose the average number of particles
on the left and on the right sides of the chains.
With some misuse of language, but sticking to standard notations, we  will call ``densities" these averages and we will denote them $\rho_a$ (left reservoir)
and $\rho_b$ (right reservoir). The  vaules of $\rho_a$ and $\rho_b$ are reported in the table below and computed in Sec.\ref{sezeq}.\\
\begin{table}[h!]
\centerline{
\begin{tabular}{|c|c|c|}
\hline System & $\rho_a$ & $\rho_b$\\
\hline  SIP & $ 2k \frac{\a}{\g-\a}$ &  $2k \frac{\d}{\b-\d}$ \\
\hline  SEP &  $2j \frac{\a}{\g+\a}$ & $2j \frac{\d}{\b+\d}$\\
\hline  IRW & $  \frac{\a}{\g}$  &  $\frac{\d}{\b}$ \\
\hline
\end{tabular}}
\caption{\label{rho-def} Definition of $\rho_a$ and $\rho_b$.}
\end{table}
}
\end{remark1}
\begin{remark1}
{Note that the SIP process requires $\gamma>\alpha$ and $\beta>\delta$.
This condition turns out to be   necessary in order for the system to reach a stationary state (see also formula \eqref{cond}).
}
\end{remark1}
\begin{remark1}\label{REM}
It is interesting to remark that the exclusion (resp. inclusion) walkers with parameters $(\alpha,\gamma',\delta,\beta')$ converges
to the independent walkers with parameters $(\alpha,\gamma,\delta,\beta)$ in the limit
$j\to \infty$ (resp. $k\to\infty$) under the scaling $\gamma' = 2j \gamma$,  $\b' = 2j \b$ (resp. $\gamma' = 2k \gamma$,  $\b' = 2k \b$).
Indeed, in this limit the generators  $\frac{{\cal L}^{SIP}}{2j}$ and $\frac{{\cal L}^{SEP}}{2j}$ converge to ${\cal L}^{IRW}$. This remark can be put on rigorous grounds by
using the Trotter-Kunz theorem (see Theorem 2.12 of \cite{L}); see for istance  \cite{GKRV} for the proof in the case of  SEP($2j$).
\end{remark1}
\subsection{Interacting diffusions}
The last process  we consider is the Brownian Energy Process (BEP), originally introduced (without boundaries) in \cite{GKRV}. Here we present its boundary driven version. The bulk diffusion process of the BEP also appears in genetics, as the multi-type Wright-Fisher diffusion with parent independent mutation rate (see \cite{CGGR}
and references therein for a discussion of duality in the context of population dynamics).

\vspace{0.5cm}
{\bf Brownian energy process BEP($2k$).}
This model describes symmetric energy exchange between nearest neighbouring
sites, see Figure \ref{figBEP}. The dynamical variables (energies) are collected in the vector
$z=(z_1,\ldots,z_L)\in \R_+^L$ and the generator is
\begin{eqnarray}\label{BEP}
{\cal L}^{BEP}&=&{\cal L}_a^{BEP}f(\eta)+{\cal L}_0^{BEP}f(\eta)+{\cal L}_b^{BEP}f(\eta)\\
&=&T_a \left (2k\frac{\partial }{\partial z_1} + z_1\frac{\partial^2}{\partial z_1^2} \right ) - \frac12 z_1\frac{\partial }{\partial z_1}\nonumber\\
&+& \sum_{i=1}^{L-1}
z_iz_{i+1} \left(\frac{\partial }{\partial z_i}-\frac{\partial }{\partial z_{i+1}}\right )^2 - 2k (z_i-z_{i+1})\left (\frac{\partial }{\partial z_i}-\frac{\partial }{\partial z_{i+1}}\right )\nonumber\\
&+&
T_b \left (2k\frac{\partial }{\partial z_L} + z_L\frac{\partial^2}{\partial z_L^2}\right ) - \frac12 z_L\frac{\partial }{\partial z_L}\;.\nonumber
\end{eqnarray}
\begin{remark1}
The origin of the bulk dynamics, generated by
\begin{eqnarray}\label{BEP0}
{\cal L}_0^{BEP}f(z)&=&
\sum_{i=1}^{L-1}
z_iz_{i+1} \left(\frac{\partial}{\partial z_i} - \frac{\partial}{\partial z_{i+1}}\right)^2
- 2k (z_i-z_{i+1})\left(\frac{\partial}{\partial z_i}-\frac{\partial}{\partial z_{i+1}}\right)
\end{eqnarray}
can be explained as follows \cite{GKR,GKRV}. Consider $m=4k\in\mathbb{N}$ velocity variables on
each site $i$ and call them $v_{i,\alpha}$ with $\alpha=1,\ldots,m$. Suppose that they evolve
with the following generator
\be\label{BMP}
{\cal L}_0^{BMP}f(v) = \sum_{i=1}^{L-1}\sum_{\alpha,\beta=1}^{m}
\left(
v_{i,\alpha}\frac{\partial}{\partial v_{i+1,\beta}}
- v_{i+1,\beta}\;\frac{\partial}{\partial v_{i,\alpha}}
\right)^2
\ee
which defines a process, called Brownian Momentum Process, introduced in \cite{BO,GK}.
Each term in ${\cal L}_0^{BMP}$ represents a rotation in the plane $(v_{i,\alpha},v_{i+1,\beta})$,
therefore it conserves the total length $v_{i,\alpha}^2 + v_{i+1,\beta}^2$,
i.e. the total kinetic energy. One can check that the BEP$(2k)$ is the evolution
process, induced by (\ref{BMP}), of the total energies on each site
\be\label{sumv}
z_i = \sum_{\alpha=1}^{m} v_{i,\alpha}^2\;.
\ee

The generator of the BEP reservoirs  ${\cal L}_a^{BEP}$ and ${\cal L}_b^{BEP}$, that will  be discussed in some details in Sec. \ref{sezeq}, impose an average energy $4k T_a$ on the left,
and an average energy $4kT_b$ on the right. The choice of their form is motivated as follows.
%
Consider an Ornstein-Uhlenbeck process on each of the $m$ velocities at site $1$ of the Brownian Momentum  process (\ref{BMP}), namely
\be
{\cal L}^{BMP}_{a} = \sum_{\alpha=1}^{m} 2T \frac{\partial^2}{\partial v_{1,\alpha}^2} - v_{1,\alpha } \frac{\partial}{\partial v_{1,\alpha}}\;.
\ee
Since in the stationary state of this reservoir the $\{v_{1,\alpha}\}_{\alpha=1,\ldots,m}$ are independent centered Gaussian
with variance $T$ then, using (\ref{sumv}), the expectation of $z_1$ is $\mathbb{E}(z_1) =
\sum_{\alpha=1}^m \mathbb{E}(v_{1,\alpha}^2) = mT = 4kT$ .
\end{remark1}
\begin{figure}[h]
\centering
\includegraphics[width=12cm]{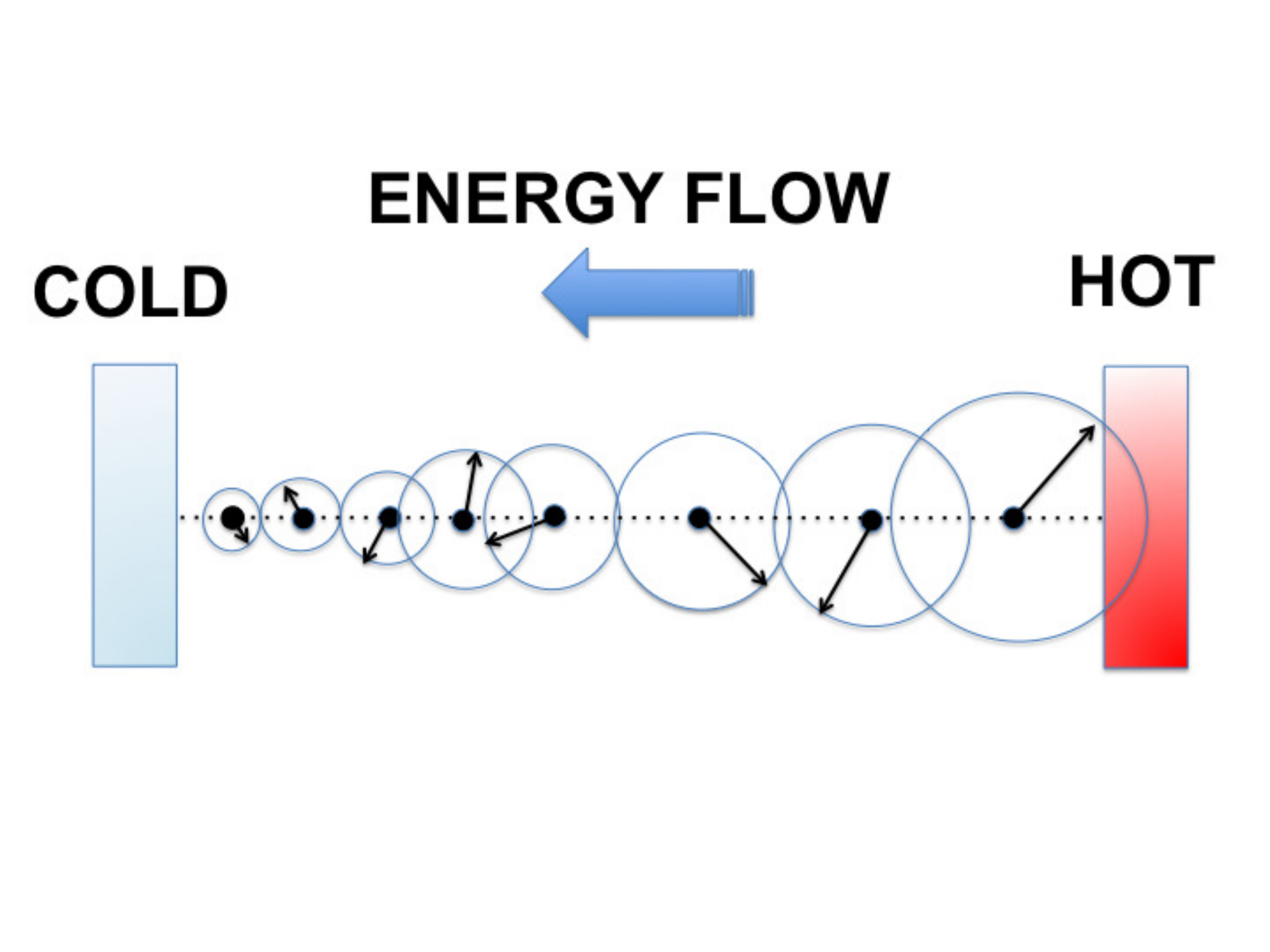}
\vspace{-2cm}
\caption{\label{figBEP} Schematic description of the Brownian Energy Process BEP($2k$). The length of the arrows represent the energies $z_i$ on the sites $i$'s, the two rectangles represent the cold and hot boundary  reserovoirs.}
\end{figure}

\subsection{Scaling limit of the particle systems}

Besides duality, there is another relation connecting the bulk part of the BEP with
generator ${\cal L}_0^{BEP}$ (\ref{BEP0}),
and the bulk part of the SIP with
generator ${\cal L}_0^{SIP}$  (third line in (\ref{INC})).
The BEP can be indeed  obtained from the SIP, through a suitable scaling limit, by a reinterpretation of this process as a model of energy transport, by supposing that each particle carries a quantum of energy $\eps$.
In this interpretation, since ${\cal L}_0^{SIP}$ conserves the number of particles, then
it conserves the total energy. Consider the free boundary inclusion process $\eta(t)=\(\eta_1(t),\dots, \eta_{L}(t)\)$ generated by  ${\cal L}_0^{SIP}$  and let $N$ be the total number of particles, i.e. $N=\sum_{i=1}^L \eta_i$.
Let $\eps$ be a parameter of the order of $1/N$, then
one expects $\eta_i$ to be of the order of $\eps^{-1}$
as $\eps \to  0$
(despite attractive interactions for any finite $k$ there are
no condensation phenomena in the SIP; one needs to rescale
$k$ with $\epsilon$ to see particles coalescing into a single site;
see \cite{GRV2}). Then  one may investigate the continuous
dynamics generated in the limit as $\eps \to 0$ on the variables
$z_i(t)=\eps \eta_i(t)$.
It turns out that the limiting dynamics for $z(t)$ is generated by
${\cal L}_0^{BEP}$.

\vskip.3cm

\begin{prop}
\label{scalingsip}
Let $\eta(t)=\(\eta_1(t),\dots, \eta_{L}(t)\)$ be the bulk inclusion process generated by  ${\cal L}_0^{SIP}$  with  $N$ particles. Let $\eps=\mathcal E /N$ for some fixed $\mathcal E>0$. Then the process $z(t)=\(z_1(t), \dots, z_L(t)\)$ where $z_i(t)=\eps \, \eta_i(t)$ is, in the limit $\eps \to 0$, the bulk Brownian energy process generated by ${\cal L}_0^{BEP}$ with total energy $\mathcal E$.
\end{prop}

\vskip.3cm
{\bf Proof.}
Let  $F:\mathbb{R}_+^L \to \mathbb{R}$, $F=F(z)$ be a two times countinuously differentiable function, i.e. $F \in \mathcal C^2(\mathbb{R}_+^L)$.
Let $z^\eps =\(z_1^\eps, \dots, z_L^\eps\) \in \mathbb{R}_+^L$ be such that $z^\eps/\eps \in \mathbb{N}_+^L$, then for any $F$ as above, there exists
 $f:\mathbb{N}_0^L \to \mathbb{R}$, $f=f(\eta)$,  such that
\begin{equation}
F\(z^\eps_1, \dots, z^\eps_L\):=f\(\frac{z^\eps_1}{\eps}, \dots, \frac{z^\eps_L}{\eps}\)\;.
\end{equation}
Let ${\cal L}_0^{\eps}$ be the generator of the process $z^\eps(t)$ induced by the SIP, then ${\cal L}_0^{\eps}$ acts on  $F=F(z^\eps)$ as follows:
\begin{eqnarray*}
\[{\cal L}_0^{\eps}F\](z^\eps)&=& \[{\cal L}_0^{SIP}f\]\(\frac {z^\eps} \eps\) \nonumber\\\\
&=&\sum_{i=1}^{L-1}\Bigg\{\frac{z^\eps_i}{\eps}\(2k+\frac{z^\eps_{i+1}}{\eps}\)
\[ f \( \frac{z^\eps_1}{\eps} , \dots, \frac{z^\eps_i}{\eps}-1, \frac{z^\eps_{i+1}}{\eps}+1, \dots, \frac{z^\eps_L}{\eps} \)-f \( \frac {z^\eps} \eps \) \] \nonumber
\\&& \hskip.6cm+ \: \frac{z^\eps_{i+1}}{\eps}\(2k+\frac{z^\eps_i}{\eps}\)\[ f \( \frac{z^\eps_1}{\eps} , \dots, \frac{z^\eps_i}{\eps}+1, \frac{z^\eps_{i+1}}{\eps}-1, \dots, \frac{z^\eps_L}{\eps} \)-f \( \frac {z^\eps} \eps \) \] \Bigg\}\nonumber\\
&=&\sum_{i=1}^{L-1}\Bigg\{\frac{z^\eps_i}{\eps}\(2k+\frac{z^\eps_{i+1}}{\eps}\)
\[ F \(z^\eps_1, \dots, z^\eps_i-\eps, z^\eps_{i+1}+\eps, \dots, z^\eps_L\)-F \(z^\eps\) \] \nonumber
\\&& \hskip.6cm+ \: \frac{z^\eps_{i+1}}{\eps}\(2k+\frac{z^\eps_i}{\eps}\)\[ F \(z^\eps_1, \dots, z^\eps_i+\eps, z^\eps_{i+1}-\eps, \dots, z^\eps_L\)-F \(z^\eps\)  \] \Bigg\}\;.\nonumber\\
\end{eqnarray*}
Suppose that $z^\eps$ converges to a finite limit $z^\eps \to z \in \mathbb{R}_+^L$ as $\eps \to 0$. Then, from the regularity assumptions on $F$, we have
\begin{eqnarray}
[\Delta^\eps_{i,i+1}F](z^\eps)&:=&F(z^\eps_1,\ldots,z^\eps_{i-1}-\eps,z^\eps_i+\eps,\ldots, z^\eps_L)-F(z^\eps_1,\ldots,z^\eps_{i-1},z^\eps_i,\ldots, z^\eps_L) \nonumber \\&=&-\eps\left (\frac{\partial }{\partial z_{i}}-\frac{\partial }{\partial z_{i+1}}\right )F(z^\eps)+{ o(\eps)}\;,
\end{eqnarray}
while
\begin{eqnarray}
 [\Delta_{i+1,i}^\eps F](z^\eps)&:=&F(z^\eps_1,\ldots,z^\eps_{i-1}+\eps,z^\eps_i-\eps,\ldots,z^\eps_L)-F(z^\eps_1,\ldots,z^\eps_{i-1},z^\eps_i,\ldots,z^\eps_L)\nonumber \\&=&\eps\left (\frac{\partial }{\partial z_{i}}-\frac{\partial }{\partial z_{i+1}}\right )F(z^\eps)+{ o(\eps)}\;,
 \end{eqnarray}
and
$$
[(\Delta^\eps_{i+1,i}+\Delta^\eps_{i,i+1})F](z^\eps)=\eps^2 \left ( \frac{\partial }{\partial z_{i}}-\frac{\partial }{\partial z_{i+1}}\right  )^2F(z^\eps)+ o(\eps^2)\;.
$$

Therefore we have
\be\label{appross1}
\[{\cal L}_0^{\eps}F\](z^\eps)=\[-2k(z^\eps_i-z^\eps_{i+1})\left(\frac{\partial}{\partial z_i}-\frac{\partial}{\partial z_{i+1}}\right) +
z^\eps_i z^\eps_{i+1}   \left ( \frac{\partial }{\partial z_{i}}-\frac{\partial }{\partial z_{i+1}}\right  )^2\]F(z^\eps) +o(1)\;.
\ee
Thus, for any $F$ as above,
 $ \lim_{\eps \to 0}[ {\cal L}_0^{\eps} F ](z^\eps)=[{\cal L}_0^{BEP}F](z)$. Moreover the total energy is clearly conserved in the limit and it is given by $\sum_{i=1}^L z_i=\sum_{i=1}^L z^\eps_i= \eps \, N =\mathcal E$.  \qed

\vskip.3cm


The same scaling analysis of the inclusion walkers can be performed on the bulk dynamics of
independent random walkers. This yields a deterministic process as scaling limit, which is
also dual to independent random walkers (cfr. \cite{GKRV}).
\begin{prop}
\label{scalingirw}
Let $\eta(t)=\(\eta_1(t),\dots, \eta_{L}(t)\)$ be the bulk process generated by  ${\cal L}_0^{IRW}$  with  $N$ particles. Let $\eps=\mathcal E /N$ for some fixed $\mathcal E>0$. Then the process $y(t)=\(y_1(t), \dots, y_L(t)\)$ where $y_i(t)=\eps \, \eta_i(t)$ is, in the limit $\eps \to 0$, the deterministic energy
process (DEP) with total energy $\sum_{i=1}^{L-1} y_i(t) = \mathcal E$ generated by
$$
{\cal L}_0^{DEP} = \sum_{i=1}^{L-1} (y_i-y_{i+1})\left(\frac{\partial}{\partial y_{i+1}}-\frac{\partial}{\partial y_{i}}\right)\;.
$$
\end{prop}

\begin{remark1}
One may wonder whether there exists a diffusion process arising as a limit of the Exclusion process.
By performing  an analogous  scaling as above, the rates of the SEP take the form $N z_i(2j-N z_{i+1})$ that become negative in the limit as $N \to \infty$.
Consistently the limit of the SEP generator is a second order differential operator
that cannot be interpreted as the generator of a Markov process, since it is has a negative coefficient in front of the second order derivatives, i.e.
\be\label{appross}
\sum_{i=1}^L \: - z_i z_{i+1} \left ( \frac{\partial }{\partial z_{i+1}}-\frac{\partial }{\partial z_{i}}\right  )^2-2j(z_i-z_{i+1})\left( \frac{\partial }{\partial z_{i}}-\frac{\partial }{\partial z_{i+1}}\right)\;.
\ee
\end{remark1}

\begin{remark1}
The same scaling limit which transfoms the bulk dynamics of the SIP into the one of the  BEP does not  work with the reservoirs. Indeed, applying to ${\cal L}^{SEP}_a$ the scaling of Propositon \ref{scalingsip}, the resulting generator is
\be
 (\a-\g)z_1 \frac{\partial}{\partial z_1}
\ee
which produces a deterministic behavior: $z_1(t) = z_1(0) e^{(\a-\g)t}$ . On the other hand it is simple to check that the thermal bath of the BEP can be obtained from
a boundary driven SIP with a modified reservoir generated by
\be
{\cal L}^{SIP}_{a,q} = \left ( 2kq+\left (q-\frac12\right )\eta_1\right)\[f(\eta^{0,1})-f(\eta)\]+ q \eta_1 \[f(\eta^{1,0})-f(\eta)\]\
\ee
with the condition $q\eps \rightarrow T_a$ as $\eps\rightarrow 0$.
\end{remark1}

\section{Stationary measures at equilibrium}\label{sezeq}
The models introduced in the previous section are Markov processes with discrete or
continuous state spaces. The long term behaviour of the processes are described by their stationary measures.
In general it is hard to determine such measures and, in fact, the invariant states of SIP, SEP and BEP
in non-equilibrium conditions are not explicitly known.
The problem of finding the explicit form of the invariant states is greatly
simplyfied at equilibrium.  The equilibrium condition for our systems can be obtained in two ways: either by suppressing the reservoirs
(i.e. considering only the bulk dynamics ${\cal L}_0$) or, retaining the reservoirs, by imposing equal densities
or equal temperatures at the boundaries of the chain, i.e. $\rho=\rho_a=\rho_b$ or $T=T_a=T_b$.\\
In the first case there exists an infinite family of reversible measures labelled by a continuous parameter.
In the second case (i.e. in the presence of the reservoir) at density $\rho$ (resp. at temperature $T$) the boundary conditions select one reversible measure.



\subsection{Equilibrium product measures}
Reversible invariant probability measure $\mathbb P$ of the bulk dynamics generated by
${\cal L}_0$ can be obtained by imposing the detailed balance condition.
When the state space $\Omega$ is finite or countable,
this condition is expressed by requiring that for any pair of configurations
$\eta,\, \eta^\prime \in \Omega $ the probability $\mathbb P$ satisfies
\be\label{bilanciodett}
L_0(\eta,\eta^\prime){\mathbb P}(\eta)=L_0(\eta^\prime,\eta){\mathbb P}(\eta^\prime)
\ee
where $L_0(\eta,\eta^\prime)$ is the  transition  rate from the configuration $\eta$ to $\eta^\prime$,
i.e. $L_0(\eta,\eta^{\prime}) = {\cal L}_0 f(\eta)$ with $f(\eta)= \delta_{\eta,\eta^{\prime}}$.
When the state space $\Omega$ is continuos, a probability measure with
density $\psi(x)$ is said to be reversible stationary measure if,
for all functions $f$ and $g$ in the domain of the generator ${\cal L}_0$,
it holds
\be
\label{bilancio-cont}
\int f(x){\cal L}_0g(x) \psi(x) dx = \int {\cal L}_0f(x) g(x) \psi(x) dx\;.
\ee
By imposing (\ref{bilanciodett}) in the case of SIP, SEP, IRW and
(\ref{bilancio-cont}) in the case of BEP and requiring the factorization
of the probability measure
one obtain the reversible measures described
in the following proposition, whose proof is left to the reader.
\begin{prop}
For the bulk processes with generator ${\cal L}_0$ defined in Sec. 2 we have\\  \\
{\bf Inclusion walkers SIP($2k$)}\\
The process with generator ${\cal L}_0^{SIP}$ has a reversible stationary measure
given by products of generalized  {\tt Negative Binomial} measures with parameters
$2k>0$ and arbitrary $0<p<1$, i.e.
\be \label{distreqsip}
{\mathbb P}(\eta)=\prod_{i=1}^L  \frac{p^{\eta_i}}{\eta_i!} \frac{\Gamma(2k+\eta_i)}{\Gamma(2k)} (1-p)^{2k}\;.
\ee
 {\bf Exclusion walkers SEP($2j$)}\\
The process with generator ${\cal L}_0^{SEP}$ has reversible stationary measure
given by products of {\tt Binomial} measures with parameters $2j\in\mathbb{N}$ and arbitrary $0<p<1$, i.e.
\be \label{distreqsep}
{\mathbb P}(\eta)=\prod_{i=1}^L \frac{\left(\frac{p}{1-p}\right)^{\eta_i}}{\eta_i!} \frac{\Gamma(2j+1)}{\Gamma(2j+1-\eta_i)} (1-p)^{2j}\;.
\ee
 {\bf Independent random walkers IRW}\\
The process with generator ${\cal L}_0^{IRW}$ has reversible stationary measure
given by products of {\tt Poisson} distribution with
arbitrary parameter $\lambda>0$ i.e.
\be \label{distreqirw}
{\mathbb P}(\eta)= \prod_{i=1}^L  \frac{\lambda^{\eta_i}}{\eta_i!} e^{-\lambda}\;.
\ee
 {\bf Brownian energy process BEP($2k$)}\\
The process with generator ${\cal L}_0^{BEP}$ has reversible measures given
by product  of {\tt Gamma} distributions with
parameters $2k>0$ and arbitrary $\theta>0$, i.e.
\begin{equation} \label{distreqbep}
{\mathbb P}(dz)= \prod_{i=1}^L   \frac{1}{(\theta)^{2k} \Gamma(2k)}  z_i^{2k-1} e^{-z_i/\theta} dz_i\;.
\end{equation}
\end{prop}
\subsection{Equilibrium product measure with reservoirs} \label{Canonical}
We recall that, in the case of  particle systems (see  Sec.\ref{systems}),
the reservoirs are modeled by birth-death processes with   creation rate $b(n)$ and annihilation rate $d(n)$,  $n$  the number of particles at the boundary. Each reservoir  has, thus,  its  own
reversible invariant probability measure, $p(n)$, which satisfies the detailed balance condition $b(n)p(n)=d(n+1)p(n+1)$. This condition can be used to compute $p(n)$. The average value of the random number $n$ (that we call density, irrespective to its value) is the quantity imposed by the reservoir to the system.\\
The effects of the reservoirs, under the equilibrium conditions, are described in the following proposition, which can easily be proved with an explicit computation.
\begin{prop}
For the processes with generator ${\cal L}$ defined in Sec. 2 we have: \\ \\
{\bf Inclusion walkers SIP($2k$)}\\
The left reservoir is modeled by the birth and death process with
rates
\be
b(n) =  \alpha(2k+n), \quad d(n) = \gamma n,\quad n \in {\mathbb N}\;.
\ee
The stationary state of this reservoir is given by a Negative Binomial measure
with parameters $2k$ and $p = \frac{\alpha}{\gamma}$.
The reservoir density is $\rho :=\langle n \rangle = 2k \frac{p}{1-p} = 2k \frac{\alpha}{\gamma-\alpha}$.
The boundary driven process with generator ${\cal L}^{SIP}$ defined in (\ref{INC}),  with parameters $\a,\, \g$ and $\b,\, \d$
such that $\a\b - \g\d =0$
(and thus $\rho_a=\rho_b$)
admits the stationary product distribution:
\begin{equation}\label{cond}
  \otimes_{i=1}^L\text{ \tt Negative-Binom}\(2k, p\) \quad \quad \text{with} \quad \quad p:=\frac{\a}{\ga} =\frac{\delta}{\beta}\; \quad \quad \text{for} \quad \a < \ga \quad \text{and}  \quad \d < \b
\end{equation}
{\bf Exclusion walkers  SEP ($2j$)}\\
The left reservoir is modeled by
\be
b(n) =  \alpha(2j-n), \quad d(n) = \gamma n,\quad  n\in\{0,1,\ldots,2j\}.
\ee
The  stationary state of this reservoir is given by a Binomial measure
with parameters $2j$ and $p = \frac{\alpha}{\gamma+\alpha}$.
The reservoir density is $\rho := \langle n \rangle = 2j p = 2j \frac{\alpha}{\gamma+\alpha}$.
The boundary driven process with generator ${\cal L}^{SEP}$ defined in (\ref{EXC}),  with parameters $\a,\, \g$ and $\b,\, \d$
such that $\a\b - \g\d =0$
(and thus $\rho_a=\rho_b$)
admits the stationary product distribution:
%
%
%
\begin{equation}
\otimes_{i=1}^L\text{ \tt Binom}\(2j, p\) \quad \quad \text{with} \quad \quad p:=\frac{\a}{\g+\a}=\frac{\d}{\b+\d}\;.
\end{equation}
{\bf Independent random walkers IRW}\\
The left reservoir has a constant birth rate
\be
b(n) =  \alpha, \quad d(n) = \gamma n,\quad n\in\mathbb{N} \;.
\ee
This reservoir imposes a Poisson measure
with parameter $\lambda= \frac{\alpha}{\gamma}$.
Therefore the density (i.e. mean number of particle)
is $\rho:=\langle n \rangle = \frac{\alpha}{\gamma}$.
If $\frac{\a}{\ga} = \frac{\d}{\b}$ the process with generator ${\cal L}^{IRW}$ defined in (\ref{IND}) admits the stationary product measure:
\begin{equation}
\otimes_{i=1}^L\text{ \tt Poisson}\(\lambda\) \quad \quad \text{with} \quad \quad \lambda:=\frac{\a}{\ga} = \frac{\d}{\b}\;.
\end{equation}
{\bf Brownian energy process BEP($2k$)}\\
In this case the generator of the left reservoir is :
\begin{eqnarray}
{\cal L}_a^{BEP}&=&
T_a \left(2k\frac{\partial}{\partial z} + y\frac{\partial^2}{\partial z^2}\right) - \frac12 z \frac{\partial}{\partial z}, \quad \quad\quad z\in \R^{+},
\end{eqnarray}
The stationary measure of this reservoir is  the Gamma distribution with parameters $2k$ and $\theta = 2T_a$.
From the properties of the Gamma distribution one has $\langle z \rangle = 4k T$.
If $T_a=T_b$ then the
process with generator ${\cal L}^{BEP}$ defined in (\ref{BEP}) admits the stationary product measure:
\begin{equation}
  \otimes_{i=1}^L\text{ \tt Gamma}\(2k, 2T\) \quad \quad \text{with} \quad \quad T:=T_a=T_b\;.
\end{equation}
\end{prop}
%

\section{Duality}\label{SECT:Dual}

When the reservoirs of our boundary driven processes work at different
parameters value so that different densities or temperatures are imposed
on the two sides, the stationary measure is in general unknown.
Remarkable exceptions are the boundary driven SEP(1), with at most
one particle per site, for which a matrix product solution is available
\cite{DEHP}, and the case of IRW, where the product structure of the
equilibrium invariant measure is preserved.

An alternative approach to characterize the stationary non-equilibrium
state is provided by duality.
{In section \ref{sec5uno} we describe duality for the processes previously
defined. Dual processes have absorbing boundaries at two extra sites
with suitable absorbing rates depending on the parameters reservoirs.
In general the duality functions are related to moments of the stationary
distribution. In section \ref{sec5due} we show several applications of duality
and we obtain via duality the stationary non-equilibrium measure of independent
random walkers.}

\subsection{{Dual processes}}
\label{sec5uno}
Consider the extended chain $\{0,1, ..., L, L+1\}$ obtained from the original one by adding the bounday sites $\{0,L+1\}$.
 Let $\eta=(\eta_1, ...,\eta_L)$ be the
configuration in the original process, we denote by $\xi=(\xi_0, \xi_1, ..., \xi_L, \xi_{L+1})\in\Omega^{dual}$ the configuration for the dual process, where the configuration space $\Omega_{Dual}$
will be specified later. We say that $(\eta_t)_{t\ge 0}$ and $(\xi_t)_{t\ge 0}$ are dual with duality function $D(\eta,\xi)$ if
\begin{equation}\label{Duality}
\mathbb E_\eta[D(\eta_t, \xi)]=\mathbb E_\xi[D(\eta, \xi_t)] \quad \quad\quad  \text{for any } \: \quad t\ge 0, \quad (\eta,\xi)\in\Omega\times\Omega_{Dual}\;,
\end{equation}
where $\mathbb E_\eta$ denotes the expectation in the original process started from the configuration $\eta$, whereas $\mathbb E_\xi$ denotes the expectation in the dual process started from the configuration $\xi$.
\begin{teorema1}
\label{mainteo}
For the processes defined in Sec. 2 we have
the following duality results.\\ \\
{\bf Inclusion walkers SIP($2k$).}
The process $(\eta_t)_{t\ge0}$ defined by \eqref{INC} is dual to the absorbing boundaries process
$(\xi_t)_{t\ge0}$ with configuration space $\Omega_{Dual}=\mathbb{N}_0^{L+2}$ with generator
\begin{eqnarray}
\label{dual-inc}
{\cal L}^{SIP}_{\text{Dual}}f(\xi)&=&(\ga-\a) \xi_1 \[f(\xi^{1,0})-f(\xi)\]\\
&+& \sum_{i=1}^{L-1}\xi_i(2k+\xi_{i+1})\[f(\xi^{i,i+1})-f(\xi)\]+  \xi_{i+1}(2k+\xi_i)\[f(\xi^{i+1,i})-f(\xi)\]\nonumber\\
&+&(\b-\d) \xi_L \[f(\xi^{L,L+1})-f(\xi)\]\;,\nonumber
\end{eqnarray}
with duality function
\begin{eqnarray}
\label{dualSIP}
D^{SIP}(\eta,\xi)=\(\frac \a {\ga-\a}\)^{\xi_0} \; \prod_{i=1}^L \frac{\eta_i !}{(\eta_i-\xi_i)!} \, \frac{\Gamma(2k)}{\Gamma(2k+\xi_i)}\; \(\frac \d {\b-\d}\)^{\xi_{L+1}}\;.
\end{eqnarray}

{\bf Exclusion walkers SEP($2j$).}
The process $(\eta_t)_{t\ge0}$ defined by \eqref{EXC} is dual to the absorbing boundaries process $(\xi_t)_{t\ge0}$ with configuration space $\Omega_{Dual}=\mathbb{N}_0\times\{0,1,\ldots,2j\}^{L}\times\mathbb{N}_0$ with generator
\begin{eqnarray}
{\cal L}^{SEP}_{\text{Dual}}f(\xi)&=&(\a+\ga) \xi_1 \[f(\xi^{1,0})-f(\xi)\]\\
&+& \sum_{i=1}^{L-1}\xi_i(2j-\xi_{i+1})\[f(\xi^{i,i+1})-f(\xi)\]+ \xi_{i+1}(2j-\xi_i)\[f(\xi^{i+1,i})-f(\xi)\]\nonumber\\
&+&(\b+\d) \xi_L \[f(\xi^{L,L+1})-f(\xi)\]\;,\nonumber
\end{eqnarray}
with duality function
\begin{equation}\label{dualSEP}
D^{SEP}(\eta,\xi)=\(\frac \a {\a+\ga}\)^{\xi_0} \; \prod_{i=1}^L \frac{\eta_i !}{(\eta_i-\xi_i)!} \, \frac{\Gamma(2j+1-\xi_i)}{\Gamma(2j+1)}\; \(\frac \d {\b+\d}\)^{\xi_{L+1}}\;.
\end{equation}

{\bf Independent random walkers IRW.}
The process $(\eta_t)_{t\ge0}$ defined by \eqref{IND} is dual to the absorbing boundaries process $(\xi_t)_{t\ge0}$ with configuration space $\Omega_{Dual}=\mathbb{N}_0^{L+2}$ with generator
\begin{eqnarray}
\label{gen-dual-ind}
{\cal L}^{IRW}_{\text{Dual}}f(\xi)&=& \ga \xi_1 \[f(\xi^{1,0})-f(\xi)\]\\
&+& \sum_{i=1}^{L-1}\xi_i\[f(\xi^{i,i+1})-f(\xi)\]+ \xi_{i+1}\[f(\xi^{i+1,i})-f(\xi)\]\nonumber\\
&+&\b \xi_L \[f(\xi^{L,L+1})-f(\xi)\]\;,\nonumber
\end{eqnarray}
with duality function
\begin{equation}\label{dualIRW}
D^{ind}(\eta,\xi)=\(\frac \a \ga\)^{\xi_0} \; \prod_{i=1}^L \frac{\eta_i !}{(\eta_i-\xi_i)!} \; \(\frac \d \b\)^{\xi_{L+1}}\;.
\end{equation}

{\bf Brownian energy process BEP($2k$).}
The process $(z_t)_{t\ge0}$ defined by \eqref{BEP} is dual to the absorbing boundary process $(\xi_t)_{t\ge0}$ with configuration space $\Omega_{Dual}=\mathbb{N}_0^{L+2}$ with generator
\begin{eqnarray}
\label{dual-bep}
{\cal L}^{BEP}_{\text{Dual}}f(\xi)&=& \frac{\xi_1}{2} \[f(\xi^{1,0})-f(\xi)\]\\
&+& \sum_{i=1}^{L-1}\xi_i(2k+\xi_{i+1})\[f(\xi^{i,i+1})-f(\xi)\]+ \xi_{i+1}(2k+\xi_i)\[f(\xi^{i+1,i})-f(\xi)\]\nonumber\\
&+& \frac{\xi_L}{2} \[f(\xi^{L,L+1})-f(\xi)\]\;,\nonumber
\end{eqnarray}
the duality function is
\begin{eqnarray}
\label{dual2}
D^{BEP}(z,\xi)=(2T_a)^{\xi_0} \; \prod_{i=1}^L z_i^{\xi_i} \, \frac{\Gamma(2k)}{\Gamma(2k+\xi_i)}\; (2T_b)^{\xi_{L+1}}\;.
\end{eqnarray}
\end{teorema1}


\vskip.4cm
Theorem \ref{mainteo} can be proven by explicit computations checking that the effect of the generator of a process on duality functions is the same as the effect of the generator of the dual process.
See \cite{GKR, GKRV} for this explicit computation and the proof of duality for the bulk process. The main novelty of
Theorem \ref{mainteo} consists in including a general class of boundary rates.
Therefore, we only  include the proof of the duality property for the boundary terms.
We treat the inclusion process, the proofs for the other processes being analogous.

\vskip.3cm
{\bf Proof of Duality for the SIP$(2k)$.}
From \cite{GKRV} we know that the free boundary inclusion process (i.e. the process generated by the operator ${\cal L}_0^{SIP}$ defined in \eqref{INC}) is self-dual with duality
function:
\begin{equation}\label{D0}
D_0^{SIP}(\eta,\xi)= \prod_{i=1}^L \frac{\eta_i !}{(\eta_i-\xi_i)!} \, \frac{\Gamma(2k)}{\Gamma(2k+\xi_i)}
\end{equation}
this means that the action of  ${\cal L}_0^{SIP}$ on $D_0^{SIP}(\cdot,\xi)$ and on $D_0^{SIP}(\eta,\cdot)$ is the same, i.e.
\begin{equation}
[{\cal L}_0^{SIP} D_0^{SIP}(\cdot,\xi)](\eta)=[{\cal L}_0^{SIP} D_0^{SIP}(\eta,\cdot)](\xi)
\end{equation}
thus, since ${\cal L}_0^{SIP}$ does not act on the $0$-th and $L+1$-th components of $\xi$, we have
\begin{equation}
[{\cal L}_0^{SIP} D^{SIP}(\cdot,\xi)](\eta)=[{\cal L}_0^{SIP} D^{SIP}(\eta,\cdot)](\xi)
\end{equation}
  It remains to verify that the actions of the operators ${\cal L}^{SIP}$ and ${\cal L}^{SIP}_{\text{Dual}}$  at the boundaries are the same on the duality function.
We verify this for the left boundary:
\begin{eqnarray}
[{\cal L}_{a}^{SIP}D^{SIP}(\cdot,\xi)](\eta)&=&\a (2k+\eta_1)\[D^{SIP}(\eta^{0,1},\xi)-D^{SIP}(\eta,\xi)\]+\ga \eta_1 \[D^{SIP}(\eta^{1,0},\xi)-D^{SIP}(\eta, \xi)\]\nonumber\\
&=& D^{SIP}(\eta,\xi)\; \frac{(\eta_1-\xi_1)!}{\eta_1!} \cdot	\Bigg\{ \a (2k+\eta_1)\[\frac{(\eta_1 +1)!}{(\eta_1+1-\xi_1)!}-\frac{\eta_1 !}{(\eta_1-\xi_1)!}\]
\nonumber \\& & \hskip4cm+ \: \ga \eta_1 \[\frac{(\eta_1 -1)!}{(\eta_1-1-\xi_1)!}-\frac{\eta_1 !}{(\eta_1-\xi_1)!}\] \Bigg\}	
\nonumber \\
&=&D^{SIP}(\eta,\xi)\; \frac{\xi_1}{(\eta_1+1-\xi_1)} \cdot	\left\{ \a (2k+\eta_1)-\gamma(\eta_1+1-\xi_1)\right\} \nonumber \\
&=&D^{SIP}(\eta,\xi)\; \frac{\xi_1}{(\eta_1+1-\xi_1)} \cdot	\left\{ \a(2k+\xi_1-1)-(\ga -\a) (\eta_1+1-\xi_1) \right\} \nonumber \\
&=&\a\,\frac{(2k+\xi_1-1)}{(\eta_1+1-\xi_1)} \, \xi_1 \, D^{SIP}(\eta,\xi)-(\ga -\a)  \, \xi_1 \, D^{SIP}(\eta,\xi) \nonumber \\
&=&(\ga-\a) \xi_1 \[D^{SIP}(\eta,\xi^{1,0})-D^{SIP}(\eta,\xi)\]=
[{\cal L}^{SIP}_{\text{Dual},a}D^{SIP}(\eta,\cdot)](\xi)
\end{eqnarray}
 We have used the notations  ${\cal L}^{SIP}_{\text{a}}$ and  ${\cal L}^{SIP}_{\text{Dual,a}}$ to denote the left boundary parts of the generators ${\cal L}^{SIP}$ and
${\cal L}^{SIP}_{\text{Dual}}$ (i.e. the first line in \eqref{INC}, resp. \eqref{dual-inc}). By an analogous computation it is possible to verify that
\begin{equation}
[{\cal L}_{b}^{SIP}D^{SIP}(\cdot,\xi)](\eta)=[{\cal L}^{SIP}_{\text{Dual},b}D^{SIP}(\eta,\cdot)](\xi)
\end{equation}
where ${\cal L}^{SIP}_{\text{b}}$ and  ${\cal L}^{SIP}_{\text{Dual,b}}$ are  the right boundary parts of the two generators.
This concludes the proof of the duality property. \qed
\vskip.3cm

\begin{remark1}
At this point one may wonder whether there exists a diffusion process dual to the SEP. All the attempts that we have  done in this direction seem to suggest that
this is not the case. On the other hand, one may extend the definition of duality at the level of the generators, i.e. we say that the operator
$\mathcal L$ is dual to the operator $\mathcal L_{\text{Dual}}$  with duality function $D(z,\eta)$ if
\begin{equation}
[\mathcal L D(z,\cdot)](\eta)= [\mathcal L_{\text{Dual}} D(\cdot,\eta)](z)\;.
\end{equation}

Notice that this definition does not require  $\mathcal L$ and $\mathcal L_{\text{Dual}}$ to be Markov generators.
Under this definition, it turns out that the SEP$(2j)$ free boundary operator $\mathcal L_0^{\text{SEP}}$ is ``dual'' to the differential operator defined in \eqref{appross} that has been obtained as a scaling limit of the SEP$(2j)$.
\end{remark1}

\subsection{{Moments and duality.}}
\label{sec5due}

In this section we provide some applications of duality.
These generalize the applications of duality considered
before in the context of the simple symmetric exclusion process or the KMP model, \cite{Spo,KMP,GKR}.
Since the dual process voids the chain, 
the problem of computing stationary expectations for
the original process is reduced to the computation of
the absorption probabilities at the boundaries of the
dual walkers.
In particular, we will see how the $n$-points correlations
are related to the  absorption probabilities at the extra sites
$0$ and $L+1$ of $n$ dual walkers.

\subsubsection{Stationary expectations and absorption probabilities}
In the following Proposition we provide a relation connecting the expectation of the duality function
and the absorption probabilities of the dual walkers.

\begin{prop}
\label{corollario-abs}
{Let $\langle\cdot\rangle_L$ denote expectation with respect to the stationary measure
of the processes defined in Section 2. Let $(\xi(t))_{t\ge0}$
denote the dual processes defined in Theorem \ref{mainteo}.
For a given $\xi\in\Omega_{Dual}$ let $|\xi| = \sum_{i=0}^{L+1}\xi_i$ and
define $a_{m}(\xi)$ the absorption probabilities of the corresponding dual walkers initialized
at $\xi$ (i.e. $\xi_i$ dual walkers start from site $i$), namely
\be
a_{m}(\xi) = \mathbb{P}(\{\xi_0(\infty)=m,\xi_{L+1}(\infty)= |\xi|-m\} \;|\; \{\xi_i(0)=\xi_i, \quad \forall i=1,\ldots, L\} )\;.
\ee
Then we have:
in the case of the boundary driven processes $SIP(2k)$, $SEP(2j)$ and $IRW$
\be
\label{duality-stationary}
\langle D(\eta,\xi)\rangle_L = \sum_{m=0}^{|\xi|} \left(c {\rho_a}\right)^{m}
\left(c {\rho_b}\right)^{|\xi|-m} a_{m}(\xi)\;,
\ee
where $c=\frac{1}{2k}$ for $SIP(2k)$ model,  $c=\frac{1}{2j}$ for $SEP(2j)$ model,
$c=1$ for $IRW$ model, and where the densities $\rho_a$ and $\rho_b$ are defined
in Table \ref{rho-def};
in the case of the boundary driven processes $BEP(2k)$
\be\label{duality-z}
\langle D(z,\xi)\rangle_L = \sum_{m=0}^{|\xi|} \left(2T_a\right)^{m}
\left(2T_b\right)^{|\xi|-m} a_{m}(\xi)\;.
\ee
}
\end{prop}
\vskip.3cm
{\bf Proof.} We prove \eqref{duality-stationary}.
Let $\mu_{L,\rho_a,\rho_b}$ be the stationary measure of the process $\eta$ with boundary densities $\rho_a$ and $\rho_b$.
From the definition of duality in \eqref{Duality} and exploiting the fact that the dual walkers are absorbed
at the boundaries, we have
\bea
\langle D(\eta,\xi)\rangle_L &=& \int D(\eta, \xi) \, \mu_{L, \rho_a, \rho_b}(d\eta)\\
&=& \lim_{t \to \infty} \mathbb E_\eta\[D(\eta_t,\xi)\]\nonumber\\
&=& \lim_{t \to \infty} \mathbb E_\xi \[D(\eta, \xi_t)\]\nonumber\\
&=&\sum_{m=0}^{|\xi|}\(c \rho_a\)^m \(c \rho_b\)^{|\xi|-m} \, \mathbb{P}_\xi\(\{\xi_0(\infty)=m,\xi_{L+1}(\infty)= |\xi|-m\} \)\, , \nonumber
\ea
where $\mathbb P_\xi$ is the probability law of the dual process $(\xi(t))_{t\ge0}$ started at $\xi$ at time zero,
and the last identity follows from the formulas of the duality functions \eqref{dualSIP}, \eqref{dualSEP}, \eqref{dualIRW}
and the definitions of the densities given in Table 1.  The proof of \eqref{duality-z} is analogous.
\qed

\subsubsection{Averages in the stationary state}
In this section we will see that all the boundary driven stochastic models considered so far
have a linear density or temperature profile i.e. the
expectations $\langle \eta_i \rangle$ or  $\langle z_i \rangle $
with respect to the stationary measure is a linear function of $i$.
This is an immediate consequence of duality since, in order  to study the
average at site $i$ in the original process, it is enough to consider
a single dual random walker started at $i$ and it is an elementary
fact that its absorption probabilities at the boundaries will be
linear in $i$. Let us see.

\vskip.3cm
For a system of size $L$, the expectations $\langle \eta_i \rangle_L$ and $\langle z_i \rangle_L$ can be written, up to a factor, as the expectations (with respect to the stationary measures of the processes $\eta(t)$ and $z(t)$) of the duality functions
$D(\eta,\xi^i)$  computed in the configuration $\xi^i$  with  $\xi^i_j= \d_{i,j}$.
Furthermore,  using Proposition \ref{corollario-abs}, they can be explicitly found as functions of the dual absorption probabilities $p_i:=a_1(\xi^i)$ and $a_0(\xi^i)=1-p_i$ ($a_m(\xi)$ as in Proposition \ref{corollario-abs}). We have
\be \label{etai}
\langle \eta_i \rangle_L = \frac 1 c \langle D(\eta, \xi^i) \rangle_L =\rho_a \, p_i + \rho_b\,(1-p_i) \quad \quad \quad \quad \quad i=1, \dots, L\;
\ee
for  SIP, SEP and IRW, with $c$ as in Proposition \ref{corollario-abs}. Moreover, denoting by $\theta_a = 4kT_a$ and $\theta_b=4kT_b$, we have
\be\label{zetai}
\langle z_i \rangle_L= 2k  \langle D(z, \xi^i) \rangle_L =\theta_a\, p_i + \theta_b\,(1-p_i) \quad \quad \quad \quad \quad i=1, \dots, L\;
\ee
for the BEP.
It remains to compute $p_i$.\\ Let $X_t$ be the random walker moving on the chain $\{0, 1, \dots, L+1\}$ as follows.
 In the bulk $X_t$  jumps to one of  the neighbouring sites with rate $1/c$ (with $c$ as in Proposition \eqref{corollario-abs}), whereas it is absorbed by the left boundary (site $0$) with rate $u$ and by the right boundary (site $L+1$)
 with rate $v$. The values of $c, u$ and  $v$ depend on the model, they are listed in Table \eqref{boh}.

\begin{table}[h!]
\centerline{
\begin{tabular}{|c|c|c|c|c|}
\hline System & $SIP$ & $SEP$ & $IRW$ & $BEP$\\
\hline   $u$ & $ {\g-\a}$    &  ${\g+\a}$   & $ \g$ & $1/{2}$ \\
\hline $v$  &  ${\b-\d}$  & ${\b+\d}$ &  ${\b}$ &  $1/{2}$ \\
\hline $c$ & $1/{2k}$ & $1/{2j}$ & $1$ & $1/{2k}$\\
\hline
\end{tabular}}
\caption{ \label{boh} Dual processes jump rates.}
\end{table}

The value  $p_i$ can then be interpreted as the probability for the walker $X_t$ started at $i$ to be absorbed by the left boundary, i.e.   $p_i= \mathbb P\(X_\infty=0 \, | \; X_0=i\)$. They verify the following system of equations:
\begin{equation}\label{pi}
\left\{
\begin{array}{ll}
p_0=1\\\\
p_1= \frac{1}{cu+1} \, p_2 + \frac{cu}{cu+1} \, p_0\\\\
p_{i-1}-2p_i +p_{i+1}=0, \quad i=2, \dots, L-1\\\\
p_L= \frac{cv}{cv+1}\, p_{L+1}+ \frac{1}{cv+1}\, p_{L-1}\\\\
p_{L+1}=0\;.
\end{array}
\right.
\end{equation}
Thus  $p_i$ is a linear function of $i$ for $1 \le i \le L$ and the solution of \eqref{pi} is given by:
\be\label{pii}
p_i=\frac{L+\frac 1 {cv}-i}{L+\frac 1 {cu}+\frac 1 {cv}-1} \quad \quad \text{for} \quad i=1, \dots, L \quad \quad \text{and} \quad p_0=1, \:\: p_{L+1}=0\;.
\ee

Hence, from \eqref{etai}, for {SIP}, {SEP}, and {IRW} we get
\be\label{Exp}
\langle \eta_i \rangle_L= \frac{\rho_a\(L+\frac 1 {cv}-i\)+ \(i+\frac 1 {cu} -1\) \rho_b}{L+\frac 1 {cu}+\frac 1 {cv}-1} \quad \quad \quad \quad \quad i=1, \dots, L\;
\ee
with $u, v$ as in Table \eqref{boh}, and $\rho_a, \rho_b$ as in Table \eqref{rho-def}.
\vskip.3cm

%
%
%

\begin{remark1}
Under a suitable rescaling of the constants tuning the annihilation rates at the
boundaries (see Remark \ref{REM}), the solutions of the exclusion and of the
inclusion walkers scale to those of the independent walkers.
\end{remark1}

\vskip.3cm

Finally, from \eqref{zetai} and \eqref{pii}, for the {BEP} we get
\be \label{ExpBEP}
\langle z_i\rangle =
\frac{\theta_a(L+4k-i) + \theta_b(i-1+4k)}{L+8k-1}\;
\quad \quad \quad \quad \quad i=1, \dots, L\;.
\ee

and, by a similar computation, we find
\be \label{ExpKMP}
\langle z_i\rangle =
\frac{T_a(2L-3-2i)+T_b(2i-1)}{2(L-2)}
\quad \quad \quad \quad \quad i=1, \dots, L\;.
\ee
for the {KMP} model.

\subsubsection{Stationary product measure for the boundary driven independent walkers}
{In the following proposition the stationary measure for the boundary driven
IRW is obtained as  an application of the duality property.}
\begin{prop}\label{IRWStSt}
{The stationary measure of the process
with generator ${\cal L}^{IRW}$ defined in \ref{IND} is the  product measure with marginals at each site $i=1, \dots, L$ given by Poisson distribution with parameter
\be
\lambda_i = \frac{\rho_a\left(L+\frac{1}{\beta} -i\right) + \rho_b\left(i-1+\frac{1}{\gamma}\right)}{L + \frac{1}{\beta} +\frac{1}{\gamma} -1}\;.
\ee
}
\end{prop}
\begin{proof}
{Since for a random variable $X$ with Poisson distribution of parameter $\lambda$
the $n^{th}$ factorial moment is given by $\mathbb{E}(X(X-1)\ldots(X-n+1)) = \lambda^n$,
to prove the proposition is enough to check the identity
\be
\label{verifica}
\langle \,\prod_{i=1}^L \frac{\eta_i!}{(\eta_i-\xi_i)!}\, \rangle_L = \prod_{i=1}^L \lambda_i^{\xi_i}\;.
\ee
To this aim consider a dual walker that starts his walk from site $i\in\{1,\ldots, L\}$.
The probability $p_i$  of its ultimate absorption at site $0$ is given by
\be
\label{absor}
p_i = \frac{L+\frac{1}{\beta} -i}{L + \frac{1}{\beta} +\frac{1}{\gamma} -1}\;
\ee
(see \eqref{pii} and Table (2)). Using formula (\ref{duality-stationary}) and observing that the absorption
probabilities of a total of $\sum_{i=1}^L \xi_i$ dual walkers, with
$\xi_i$ of them initialized at site $i$, completely factorize because the
walkers are independent, one has
\begin{eqnarray*}
\langle \,\prod_{i=1}^L \frac{\eta_i!}{(\eta_i-\xi_i)!}\, \rangle_L
& = &
\prod_{i=1}^L \sum_{m_i=0}^{\xi_i} \rho_a^{m_i}\rho_b^{\xi_i-m_i}
 {\xi_i \choose m_i}  p_i^{m_i} (1-p_i)^{\xi_i-m_i}\\
& = &
\prod_{i=1}^L \left(\rho_a p_i + \rho_b(1-p_i)\right)^{\xi_i}\;.
\end{eqnarray*}
Inserting (\ref{absor}) in the above formula and remembering
the definition of the $\lambda_i$, equation (\ref{verifica}) is verified
and the proof of the proposition is completed.}
\end{proof}

\subsubsection{Duality moment functions}
It turns out from the previous section that the expectations of the duality functions $D(\eta_t,\xi)$ with respect to the probability law
of the original process $\eta_t$, i.e. the {\em ``duality moment  functions''}
\begin{equation}\label{G1}
G(\eta,\xi,t):= \mathbb E_\eta\[D(\eta_t,\xi)\]
\end{equation}
 are usually some kind of moments of the original process $\eta_t$ labelled by the discrete parameter $\xi\in\Omega_{dual}$.
In the case of  SEP, SIP and IRW, the function $G(\eta,\xi,t)$ is, up to a multiplicative constant depending on $\xi$, the $\xi$-th factorial moment at time $t$
when the initial value is $\eta$. In the case of BEP, the function $G(z,\xi,t):= \mathbb E_z [D(z_t,\xi)]$ is the standard $\xi$-th moment.
Under suitable conditions, the set of moments, obtained on varying the parameter $\xi$,  completely characterizes
the law of the original process.
 From duality we find that  the equations for the functions $G(\eta,\xi,t)$ are closed and quite simple to write.

 \begin{prop}\label{prop:G}
 Let $\eta_t$ and $\xi_t$ be two dual Markov processes with duality function $D(\eta,\xi)$ and let $\mathcal L$ and $\mathcal L_{\text{Dual}}$ be their generators, then the duality moment function $G(\eta,\xi,t)$ defined in \eqref{G1}  satisfies the following equation:
\begin{equation}\label{eqG}
\frac d {dt} \, G(\eta,\xi,t)= \[\mathcal L_{\text{Dual}} G(\eta, \cdot, t)\](\xi)\;.
\end{equation}
 \end{prop}
 \vskip.3cm
 {\bf Proof.}
{For any function $f=f(\eta)$ we have
\begin{equation}
\label{aho}
\frac d {dt}  \mathbb{E}_{\eta}\[f(\eta_t)\]=  \mathbb{E}_{\eta}[\mathcal L f(\eta_t)]\;.
\end{equation}
Given $\xi\in\Omega_{dual}$, applying (\ref{aho}) to $f(\eta) = D(\eta,\xi)$
and using duality,  namely $[\mathcal L D(\cdot, \xi)](\eta)=[\mathcal L_{\text{Dual}}D(\eta, \cdot)](\xi)$
one has
\begin{equation}
\frac d {dt} \mathbb E_\eta\[ D(\eta_t,\xi)\]=\mathbb E_\eta \[[\mathcal L_{\text{Dual}}D(\eta_t, \cdot)](\xi)\] = [\mathcal L_{\text{Dual}} \, \mathbb E_\eta \[D(\eta_t, \cdot)]\](\xi).
\end{equation}
Equation \eqref{eqG} follows from the definition of the function $G$ (cfr. (\ref{G1})). }\qed

\begin{coro}\label{coro:G}
{Let $\langle\cdot\rangle$ denote expectation in the stationary state and define
the {\em ``stationary duality moment  functions''}
\begin{equation}\label{G0}
G(\xi):= \langle D(\eta, \xi) \rangle\;.
\end{equation}
It immediately follows from Proposition \ref{prop:G} that $G(\xi)$ satisfies the equation
\begin{equation}\label{eq:G0}
\mathcal (L_{\text{Dual}}G)(\xi)=0\;.
\end{equation}}
\end{coro}

{We will see an application of the function $G$ in section \ref{sec-L1}.}

\section{Instantaneous thermalization and KMP model}
\label{Therma}

In this Section we define the boundary driven process with
instantaneous thermalization. An {\em instantaneous thermalization}
process gives rise, for each couple of nearest neighbouring sites,
to an instantaneous redistribution of the total energy
(or of the total number of particles).
The class of instantaneous thermalization processes we consider
in this paper is obtained from the non-equilibrium processes defined
so far after performing a suitable ``instantaneous thermalization limit'':
for each bond, the total energy $E$ (or the total number of particles)
of that bond is redistributed according to the stationary measure of
the original process at equilibrium on that bond, conditioned to the
conservation of $E$.

\subsection{{Thermalized models}}
To start with we  recall a well known instantaneous thermalization
model, the KMP model (see \cite{KMP}).
The KMP model is defined by considering
on each bond a  uniform redistribution of energy. At the boundaries the energy is
fixed by a reservoir which imposes a Boltzmann-Gibbs exponential energy
distribution with different temperatures $T_a$ and $T_b$.
The generator of the process is
\begin{eqnarray}
{\cal L}^{KMP}f(z)
& = &
\int_0^{\infty} dz_1'\; \frac{e^{-z_1'/T_a}}{T_a}\left(f(z'_1,z_2,\ldots,z_L) - f(z)\right) \\
& + &
\sum_{i=1}^{L-1} \int_{0}^1 dx \left(f(z_1,\ldots, x(z_i+z_{i+1}), (1-x)(z_i+z_{i+1}),\ldots,z_L) - f(z)\right)
\nonumber\\
& + &
\int_0^{\infty} dz_L' \; \frac{e^{-z_L'/T_b}}{T_b}\left(f(z_1,\ldots,z_{L-1},z'_L) - f(z)\right) \nonumber
\end{eqnarray}
for any $f:{\R}_+^L\rightarrow \R$.

At the end of this section we will see that the KMP model can be obtained as the
instantaneous thermalization limit of the BEP$(2k)$ model in the particular case $k=1/2$.

From \cite{KMP} we know that the KMP is dual to a suitable  discrete Markov process.
The dual process $\xi(t)=(\xi_0(t), \xi_1(t), \dots, \xi_L(t), \xi_{L+1}(t))\in\mathbb{N}_0^{L+2}$
describes the motion of particles in a one dimensional $L+2$ -sites chain.
The boundary sites $\xi_0$ and $\xi_{L+1}$ are absorbing.  In the bulk, for each couple of neighbouring sites $(i,i+1)$ there is an instantaneous uniform
 redistribution of the total number of particles $\xi_i+\xi_{i+1}$. The redistribution takes place whenever an exponentially distributed clock rings.
 The clocks (one for each couple $(i,i+1)$) are mutually independent.
The generator of this process is defined on functions $f: {\mathbb N}_0^{L+2}\rightarrow \R$ by
\begin{eqnarray}\label{KMP-Dual}
{\cal L}^{KMP}_{\text{Dual}}f(\xi)
& = &
\[f(\xi_0+\xi_1, 0, \xi_2, \dots, \xi_{L+1})-f(\xi)\] \\
& + &
\sum_{i=1}^{L-1}   \sum_{r=0}^{\xi_i+\xi_{i+1}}  \[f(\xi_0,\ldots,\xi_{i-1}, r,  \xi_i+\xi_{i+1}-r,\ldots,\xi_{L+1}) - f(\xi)\]
\nonumber\\
& + &
\[f(\xi_0, \dots,0 , \xi_L+\xi_{L+1})-f(\xi)\]\;, \nonumber
\end{eqnarray}

and the duality function is $D^{\text{KMP}}(z,\xi)= T_a^{\xi_0}\, \prod_{i=1}^N \frac{z_i^{\xi_i}}{\xi_i !}\; T_b^{\xi_{L+1}}$.

\vskip.3cm
We will see that, for each of the instantaneous thermalization processes that we are going to introduce there is a dual process.
The dual processes are instantaneous thermalization processes themselves. They have absorbing boundaries and can be naturally
derived by a thermalization limit from the dual processes of the original ones (see Section \ref{SECT:Dual}).



\vspace{0.4cm}
{{\bf Thermalized Inclusion walkers Th-SIP($2k$).}}
The instantaneous thermalization limit of the Inclusion process  is obtained as follows.
Imagine on each bond $(i,i+1)$ to run the SIP$(2k)$ dynamics for an infinite amount of time.
Then the total number of particles   on the bond will be redistributed according to the stationary
measure on that bond, conditioned to conservation of the total number of particles of the bond.
We consider two independent random variables  $\eta_i$ and $\eta_{i+1}$ distributed according to the stationary measure of the SIP($2k$) at the equilibrium. Thus  $\eta_i$ and $\eta_{i+1}$ are two Negative Binomial random variables
of parameters $2k$ and $p$. Hence $\eta_i+\eta_{i+1} $ is again a Negative Binomial r.v. with parameters $4k$ and $p$
and then the distribution of one of them, given that the sum is fixed to $\eta_i + \eta_{i+1} =E$, has  a Negative Hypergeometric probability
density of parameters $(E,4k-1,2k)$, i.e.
\begin{equation}
\nu^{SIP}_{2k}(r\, | \, E):=\mathbb{P}(\eta_1=r\: |\:\eta_i+\eta_{i+1} =E) = \frac{\binom{2k+r-1}{r} \cdot \binom{2k+E-r-1}{E-r}}{\binom{4k+E-1}{E}}\;.
\end{equation}
On the other hand, the stationary
distribution of the left Inclusion reservoir is the Negative Binomial with parameters $2k$ and $\frac \a \ga$ (resp. $\frac \delta \beta$ for the right reservoir). Then
the generator of the instantaneous thermalization limit of the Inclusion
 process with reservoirs can be defined as follows
\begin{eqnarray}
{\cal L}^{SIP}_{th}f(\eta)
& = & \sum_{r=0}^\infty \[f(r,\eta_2, \dots, \eta_L)-f(\eta)\]\; \binom{2k+r-1}{r}\, \(\frac \a \ga\)^r \(\frac {\ga-\a} \ga\)^{2k} \nonumber
\\
& &\hskip-2.2cm
+\,\sum_{i=1}^{L-1} \sum_{r=0}^{\eta_i+\eta_{i+1}}\[f(\eta_1, \dots, \eta_{i-1},r, \eta_i+\eta_{i+1}-r, \eta_{i+2}, \dots, \eta_L)-f(\eta)\]\;\nu^{SIP}_{2k}(r\, | \,\eta_i+\eta_{i+1})
\nonumber\\
& + & \sum_{r=0}^\infty \[f(\eta_1, \dots, \eta_{L-1},r)-f(\eta)\]\; \binom{2k+r-1}{r}\, \(\frac \delta \beta\)^r \(\frac {\b-\delta} \beta\)^{2k}\;.
\end{eqnarray}

It is easy to check that the thermalized inclusion process is dual, with duality function \eqref{dualSIP}, to the process
that behaves in the bulk as the thermalized SIP$(2k)$ itself, and which has absorbing boundaries at
two extra sites with  absorbing rate 1. In other words the dual process is generated by:
\begin{eqnarray}\label{SIPthDual}
{\cal L}^{SIP}_{th, \text{Dual}}f(\xi)
& = &  \[f(\xi_0+\xi_1,0,\xi_2, \dots,  \xi_{L+1})-f(\xi)\] \nonumber
\\
& &\hskip-2.2cm
+\,
\sum_{i=1}^{L-1} \sum_{r=0}^{\xi_i+\xi_{i+1}}\[f(\xi_0, \dots, \xi_{i-1},r, \xi_i+\xi_{i+1}-r, \xi_{i+2}, \dots, \xi_{L+1})-f(\xi)\]\;\nu^{SIP}_{2k}(r\, | \,\xi_i+\xi_{i+1})
\nonumber\\
& + & \[f(\xi_0, \dots, \xi_{L-1},0, \xi_{L}+\xi_{L+1})-f(\xi)\]\;,
\end{eqnarray}
where $\xi=\(\xi_0, \xi_1, \dots, \xi_L, \xi_{L+1}\)$.

\vspace{0.4cm}
{{\bf Thermalized Exclusion walkers Th-SEP($2k$).}}
If we take two independent random variables  $\eta_i$ and $\eta_{i+1}$ with  Binomial distribution
of parameters $2j$ and $p$, then $\eta_i+\eta_{i+1} $ is again a  Binomial r.v. with parameters $4j$ and $p$;
then the distribution of one of them, given the sum  fixed to $\eta_i + \eta_{i+1} =E$, has an Hypergeometric distribution with
parameters $(E,4j,2j)$, i.e. a probability mass function
\begin{equation}
\nu^{SEP}_{2j}(r\, | \, E):=\mathbb{P}(\eta_1=r\: |\:\eta_1+\eta_2 =E) = \frac{\binom{2j}{r} \cdot \binom{2j}{E-r}}{\binom{4j}{E}} \; \mathbf{1}_{r \le 2j}\;.
\end{equation}
The stationary
distribution of the left Exclusion reservoir is the  Binomial with parameters $2j$ and $\frac \a {\ga+\a}$ (resp. $\frac \delta {\beta+\d}$ for the right reservoir). Then
we define the generator of the instantaneous thermalization limit of the Exclusion
 process with reservoirs  as follows
\begin{eqnarray}
{\cal L}^{SEP}_{th}f(\eta)
& = & \sum_{r=0}^{2j} \[f(r,\eta_2, \dots, \eta_L)-f(\eta)\]\; \binom{2j}{r}\, \(\frac \a {\ga}\)^r \(\frac \ga {\ga+\a}\)^{2j} \nonumber
\\
& &\hskip-2.2cm
+\,
\sum_{i=1}^{L-1} \;\sum_{r=0}^{\eta_i+\eta_{i+1}}\[f(\eta_1, \dots, \eta_{i-1},r, \eta_i+\eta_{i+1}-r, \eta_{i+2}, \dots, \eta_L)-f(\eta)\]\;\nu^{SEP}_{2j}(r\, | \, \eta_i+\eta_{i+1})
\nonumber\\
& + & \sum_{r=0}^{2j} \[f(\eta_1, \dots, \eta_{L-1},r)-f(\eta)\]\; \binom{2j}{r}\, \(\frac \d {\b}\)^r \( \frac \b {\b+\d}\)^{2j}\;.
\end{eqnarray}

The thermalized exclusion process is dual, with duality function \eqref{dualSEP}, to the process
that behaves in the bulk as the process itself, and which has absorbing boundaries at
two extra sites with  absorbing rate 1:
\begin{eqnarray}\label{SEPthDual}
{\cal L}^{SEP}_{th, \text{Dual}}f(\xi)
& = &  \[f(\xi_0+\xi_1,0,\xi_2, \dots,  \xi_{L+1})-f(\xi)\] \nonumber
\\
& &\hskip-2.6cm
+\,
\sum_{i=1}^{L-1} \;\sum_{r=0}^{\xi_i+\xi_{i+1}}\[f(\xi_1, \dots, \xi_{i-1},r, \xi_i+\xi_{i+1}-r, \xi_{i+2}, \dots, \xi_L)-f(\xi)\]\;\nu^{SEP}_{2j}(r\, | \, \xi_i+\xi_{i+1})
\nonumber\\
& + & \[f(\xi_0, \dots, \xi_{L-1},0, \xi_{L}+\xi_{L+1})-f(\xi)\]\;,
\end{eqnarray}
where $\xi=\(\xi_0, \xi_1, \dots, \xi_L, \xi_{L+1}\)$.

\vspace{0.4cm}
{{\bf Thermalized Indepent walkers Th-IRW.}}
Let $\eta_i$ and $\eta_{i+1}$ be two independent random variables  with Poisson distribution
of parameter  $\lambda$, then $\eta_i+\eta_{i+1} $ is again a  Poisson r.v. with parameter $2\lambda$
then the distribution of one of them, given the sum  fixed to $\eta_i + \eta_{i+1} =E$, has
a Binomial density of parameters $(E, 1/2)$:
\begin{equation}
\nu^{IRW}(r\, | \, E):=\mathbb{P}(\eta_1=r\: |\:\eta_1+\eta_2 =E) = \binom{E}{r} \, \frac{1}{2^E}\;.
\end{equation}
Moreover the stationary
distribution of the left IRW reservoir is the Poisson with parameter  $\frac \a \ga$ (resp. $\frac \delta \beta$ for the right reservoir). Then
the generator of the instantaneous thermalization limit of the independent walkers
 process with reservoirs is given by
\begin{eqnarray}
{\cal L}^{IRW}_{th}f(\eta)
& = & \sum_{r=0}^\infty \[f(r,\eta_2, \dots, \eta_L)-f(\eta)\]\;  \(\frac{\a}{\ga}\)^r \, \frac {e^{-\a/\ga}}{r!}
\nonumber \\
& &\hskip-2.2cm
+\,
\sum_{i=1}^{L-1} \sum_{r=0}^{\eta_i+\eta_{i+1}}\[f(\eta_1, \dots, \eta_{i-1},r, \eta_i+\eta_{i+1}-r, \eta_{i+2}, \dots, \eta_L)-f(\eta)\]\;\nu^{IRW}(r\, | \, \eta_i+\eta_{i+1})
\nonumber\\
& + & \sum_{r=0}^\infty \[f(\eta_1, \dots, \eta_{L-1},r)-f(\eta)\]\;   \(\frac{\d}{\beta}\)^r \, \frac {e^{-\delta/\b}}{r!} \;.
\end{eqnarray}

The thermalized independent walkers process is dual, with duality function  \eqref{dualIRW}, to the process
that behaves in the bulk as the process itself, and which has absorbing boundaries at
two extra sites:
\begin{eqnarray}\label{IRWthDual}
{\cal L}^{IRW}_{th, \text{Dual}}f(\xi)
& = &  \[f(\xi_0+\xi_1,0,\xi_2, \dots,  \xi_{L+1})-f(\xi)\] \nonumber
\\
& &\hskip-2.2cm
+\,
\sum_{i=1}^{L-1} \sum_{r=0}^{\eta_i+\eta_{i+1}}\[f(\xi_1, \dots, \xi_{i-1},r, \xi_i+\xi_{i+1}-r, \xi_{i+2}, \dots, \xi_L)-f(\xi)\]\;\nu^{IRW}(r\, | \, \xi_i+\xi_{i+1})
\nonumber\\
& + & \[f(\xi_0, \dots, \xi_{L-1},0, \xi_{L}+\xi_{L+1})-f(\xi)\]\;,
\end{eqnarray}
where $\xi=\(\xi_0, \xi_1, \dots, \xi_L, \xi_{L+1}\)$.

\vspace{0.4cm}
{{\bf Thermalized Brownian energy process  Th-BEP($2k$).}}
We define the instantaneous thermalization limit of the Brownian Energy process  as follows.
On each bond we run the BEP$(2k)$ for an infinite time.
Then the energies on the bond will be redistributed according to the stationary
measure on that bond, conditioned to the conservation of the total energy of the bond.
If we take two independent random variables  $z_i$ and $z_{i+1}$ with Gamma distribution
of parameters $2k$ and $\theta$,
then the distribution of one of them, given  the sum  fixed to $z_i + z_{i+1} =E$, has
density
\begin{equation}
p(z_i|z_i+z_{i+1} =E) = \frac{z_i^{2k-1}(E-z_i)^{2k-1}}{\int_{0}^E z_i^{2k-1}(E-z_i)^{2k-1} dz_i}\;.
\end{equation}
Equivalently, the variable $x = z_i/E$ has a Beta$(2k,2k)$ distribution.
Denoting by $\nu_{2k}^{BEP}(x)$ the density of such  a random variable,  we can define
the generator of the instantaneous thermalization limit of the Brownian
Energy process with reservoirs as follows
\begin{eqnarray}
{\cal L}^{BEP}_{th}f(z)
& = &
\int_0^{\infty} dz_1' \;\frac{1}{(T_a)^{2k} \Gamma(2k)}  (z_1')^{2k-1} e^{-z_1'/T_a} \left(f(z'_1,z_2,\ldots,z_L) - f(z)\right) \\
& + &
\sum_{i=1}^{L-1} \int_{0}^1 \[f(z_1,\ldots, x(z_i+z_{i+1}), (1-x)(z_i+z_{i+1}),\ldots,z_L) - f(z)\]\, \nu_{2k}^{BEP}(x) dx
\nonumber\\
& + &
\int_0^{\infty} dz_L' \; \frac{1}{(T_b)^{2k} \Gamma(2k)}  (z_L')^{2k-1} e^{-z_L'/T_b} \left(f(z_1,\ldots,z_{L-1},z'_L) - f(z)\right) \;.\nonumber
\end{eqnarray}

\vskip.2cm

\begin{remark1}
 For $k=1/2$
this reproduces the uniform redistribution
rule of the KMP model on each bond of the bulk. The same is true for the reservoirs since the stationary
distribution of the Brownian Energy reservoir is Gamma with parameters $2k$ and $\theta$.
If one takes $k=1/2$, then one obtains an Exponential distribution with parameter $\theta$. \\
\end{remark1}

%


The thermalized Brownian Energy  process is dual, with duality function \eqref{dual2} to the process
that behaves in the bulk as the thermalized inclusion process, and which has absorbing boundaries at
two extra sites with  absorbing rate 1. In other words the dual process is generated by:
\begin{eqnarray}\label{BEPthDual}
{\cal L}^{BEP}_{th, \text{Dual}}f(\xi)
& = &  \[f(\xi_0+\xi_1,0,\xi_2, \dots,  \xi_{L+1})-f(\xi)\] \nonumber
\\
& &\hskip-2.2cm
+\,
\sum_{i=1}^{L-1} \sum_{r=0}^{\xi_i+\xi_{i+1}}\[f(\xi_0, \dots, \xi_{i-1},r, \xi_i+\xi_{i+1}-r, \xi_{i+2}, \dots, \xi_{L+1})-f(\xi)\]\;\nu^{SIP}_{2k}(r\, | \,\xi_i+\xi_{i+1})
\nonumber\\
& + & \[f(\xi_0, \dots, \xi_{L-1},0, \xi_{L}+\xi_{L+1})-f(\xi)\]\;,
\end{eqnarray}
where $\xi=\(\xi_0, \xi_1, \dots, \xi_L, \xi_{L+1}\)$.

\vskip.2cm

\begin{remark1}
It is immediately seen that for $k=1/2$,  this gives
the KMP-dual process defined in \eqref{KMP-Dual}.
\end{remark1}


\subsection{Stationary measures for $L=1$.}
\label{sec-L1}

In this section we compute the moments of the istantaneous thermalization  processes defined so far,
by using the result obtained in Corollary \ref{coro:G}. When $L=1$ there is no bulk  contribution in the generator, since we have a site interacting
with two sources.

The {$L=1$} case is trivial for our original processes (SEP, SIP, IRW and BEP), because for these models the contributions of the baths are additive. Then
the system is indeed equivalent to the system of one site  interacting with one bath whose parameters  are recombinations of the parameters of the two {original baths}.
The stationary measure is, then, the stationary measure of this total bath.

The interest of the {$L=1$} case for the thermalized processes lies in the fact that for these models the baths contribution are no longer additive.
Thus, even in this basic case the computation of the stationary measure is worth to be investigated.
{A result in this direction has been obtained in \cite{BDGJL} (see also the remark
at the end of this section).}

\vspace{0.4cm}
{\bf Thermalized Inclusion walkers Th-SIP($2k$)}.
From the previous section we know that   the  thermalized inclusion process $\eta_t$ is dual to the process
defined in \eqref{SIPthDual}, with duality function \eqref{dualSIP}, then the duality moment function is
\begin{equation}
G_{th}^{SIP}(\xi)={\langle D^{SIP}(\eta,\xi)\rangle}=\(\frac \a {\ga-\a}\)^{\xi_0} \; \frac{\Gamma(2k)}{\Gamma(2k+\xi_1)}\; \(\frac \d {\b-\d}\)^{\xi_{2}}\;M(\xi_1)\;,
\end{equation}
where $M(\xi)$ is the $\xi$-th factorial moment with respect to the stationary measure
of $\eta_1$:
\begin{equation}
M(\xi_1)={\langle\frac{\eta_1 !}{(\eta_1-\xi_1)!} \rangle}\;.
\end{equation}
In order to find $M(\xi)$, from Corollary \ref{coro:G} we impose
\begin{eqnarray*}
0&=&{\cal L}^{SIP}_{th, \text{Dual}}G(\xi)=G\(\xi_0+\xi_1,0,\xi_2\)-2G\(\xi_0,\xi_1,\xi_2\)+ G\(\xi_0, 0, \xi_1+\xi_2\)
\\ &= &\(\frac \a {\ga-\a}\)^{\xi_0} \;  \(\frac \d {\b-\d}\)^{\xi_{2}}  \cdot \left\{
\(\frac \a {\ga-\a}\)^{\xi_1}  -2M(\xi_1) \,  \frac{\Gamma(2k)}{\Gamma(2k+\xi_1)}+   \(\frac \d {\b-\d}\)^{\xi_1}\right\}\;.
\end{eqnarray*}
This yields
\begin{equation}
M(\xi_1)= \frac{\Gamma(2k +\xi_1)}{2\Gamma(2k)}  \[\(\frac{\a}{\ga-\a}\)^{\xi_1}+ \(\frac{\d}{\b-\d}\)^{\xi_1}\]\;.
\end{equation}

\vspace{0.4cm}
{\bf Thermalized Exclusion walkers Th-SEP($2j$)}.
The  thermalized inclusion process $\eta_t$ is dual to the process
 in \eqref{SEPthDual}, with duality function \eqref{dualSEP}, then the stationary
 duality moment function is
\begin{equation}
G_{th}^{SEP}(\xi)={\langle D^{SEP}(\eta,\xi)\rangle}=\(\frac \a {\ga+\a}\)^{\xi_0} \; \frac{\Gamma(2j+1-\xi_1)}{\Gamma(2j+1)}\; \(\frac \d {\b+\d}\)^{\xi_{2}}\; M(\xi_1)\;,
\end{equation}
where $M(\xi)$ is the $\xi$-th factorial moment with respect to the stationary measure $\nu$ of $\eta_1$:
\begin{equation}
M(\xi_1)={\langle\frac{\eta_1 !}{(\eta_1-\xi_1)!} \rangle}\;.
\end{equation}
From Corollary \ref{coro:G}, we have
\begin{eqnarray*}
0&=&{\cal L}^{SEP}_{th, \text{Dual}}G(\xi)=G\(\xi_0+\xi_1,0,\xi_2\)-2G\(\xi_0,\xi_1,\xi_2\)+ G\(\xi_0, 0, \xi_1+\xi_2\)
\\ &= &\(\frac \a {\ga+\a}\)^{\xi_0} \;  \(\frac \d {\b+\d}\)^{\xi_{2}}  \cdot \left\{
\(\frac \a {\ga+\a}\)^{\xi_1}  -2M(\xi_1) \,  \frac{\Gamma(2j+1-\xi_1)}{\Gamma(2j+1)}+   \(\frac \d {\b+\d}\)^{\xi_1}\right\}\;.
\end{eqnarray*}
This yields
\begin{equation}
M(\xi_1)= \frac{\Gamma(2j +1)}{2\Gamma(2j+1-\xi_1)} \cdot \[\(\frac{\a}{\ga+\a}\)^{\xi_1}+ \(\frac{\d}{\b+\d}\)^{\xi_1}\]\;.
\end{equation}

\vspace{0.4cm}
{\bf Thermalized independent  random walkers Th-IRW}.
The  thermalized IRW process $\eta_t$ is dual to the process
defined in \eqref{IRWthDual}, with duality function \eqref{dualIRW}, then the
stationary duality moment function is
\begin{equation}
G_{th}^{IRW}(\xi)={\langle D^{IRW}(\eta,\xi)\rangle}=\(\frac \a {\ga}\)^{\xi_0} \;  \(\frac \d {\b}\)^{\xi_{2}}\; M(\xi_1)\;,
\end{equation}
where $M(\xi_1)$ is the $\xi$-th factorial moment with respect to the stationary measure $\eta_1$ as above.
To find $M(\xi_1)$ we impose
\begin{eqnarray*}
0&=&{\cal L}^{IRW}_{th, \text{Dual}}G(\xi)=G\(\xi_0+\xi_1,0,\xi_2\)-2G\(\xi_0,\xi_1,\xi_2\)+ G\(\xi_0, 0, \xi_1+\xi_2\)
\\ &= &\(\frac \a {\ga}\)^{\xi_0} \;  \(\frac \d {\b}\)^{\xi_{2}}  \cdot \left\{
\(\frac \a {\ga}\)^{\xi_1}  -2M(\xi_1) +   \(\frac \d {\b}\)^{\xi_1}\right\}\;.
\end{eqnarray*}
This gives
\begin{equation}
M(\xi_1)= \frac 1 2 \[\(\frac{\a}{\ga}\)^{\xi_1}+ \(\frac{\d}{\b}\)^{\xi_1}\]\;.
\end{equation}

\vspace{0.4cm}
{\bf Thermalized brownian energy process Th-BEP($2k$)}.
The  thermalized brownian energy process $z_t$ is dual to the process
defined in \eqref{BEPthDual}, with duality function \eqref{dual2}, then the stationary
duality moment function is
\begin{equation}
G_{th}^{BEP}(\xi)={\langle D^{BEP}(\eta,\xi)\rangle}=\(2T_a\)^{\xi_0} \; \frac{\Gamma(2k)}{\Gamma(2k+\xi_1)}\; \(2 T_b\)^{\xi_{2}}\; M(\xi_1)
\end{equation}
where $M(\xi_1)$ is now the $\xi_1^{th}$  moment with respect to the stationary measure
of $z_1$:
\begin{equation}
M(\xi_1)={\langle z_1^{\xi_1} \rangle}\;.
\end{equation}
In order to find $M(\xi_1)$, from Corollary \ref{coro:G} we impose
\begin{eqnarray*}
0&=&{\cal L}^{BEP}_{th, \text{Dual}}G(\xi)=G\(\xi_0+\xi_1,0,\xi_2\)-2G\(\xi_0,\xi_1,\xi_2\)+ G\(\xi_0, 0, \xi_1+\xi_2\)
\\ &= &\(2T_a\)^{\xi_0} \;  \(2T_b\)^{\xi_{2}}  \cdot \left\{
\(2T_a\)^{\xi_1}  -2M(\xi_1) \,  \frac{\Gamma(2k)}{\Gamma(2k+\xi_1)}+   \(2T_b\)^{\xi_1}\right\}
\end{eqnarray*}
This yields
\begin{equation}
M(\xi_1)= \frac{\Gamma(2k +\xi_1)}{2\Gamma(2k)} \cdot \[\(2T_a\)^{\xi_1}+ \(2T_b\)^{\xi_1}\]
\end{equation}

\begin{remark1} For $k=1/2$ the $M(\xi_1)$ above becomes
\begin{equation}
M(\xi_1)= \frac{\xi_1!} 2 \cdot \[\(2T_a\)^{\xi_1}+ \(2T_b\)^{\xi_1}\]\;.
\end{equation}
{The knowledge of all the moments fully describes the stationary measure
of the KMP process with $1$ particle.
A similar result was obtained in \cite{BDGJL}. In that paper an explicit
form of the stationary measure for 1 particle connected to two reservoirs
is given, however for a process which is
slightly different that the original KMP process.  The difference lies at the
boundaries thermalization mechanism: in the KMP model the first and last sites are
directly thermalized by the reservoirs, in \cite{BDGJL}  the first and last sites
share uniformly their energies with thermalized reservoirs. }
\end{remark1}

\section{Correlations in the stationary state}
\label{compute}


For some of our processes,
such as for the SEP$(1)$ model \cite{Spo}, the BEP(1/2)
model \cite{GKR}, the KMP model \cite{BDGJL}, the covariances have been proven to be bilinear.
For the boundary driven SEP$(1)$, from the exact solution of the
microscopic stationary state (see e.g. \cite{DE,DEHP,DLS1}),
we know even more.
Indeed for this process all the correlations
$\langle \eta_{i_1}  \dots \,\eta_{i_n} \rangle$ in the stationary state
are multilinear in the variables $i_1, \dots, i_n$, and can be explicitely
computed through a recursive argument on $n$ and $L$ by a matrix
method.

The algebraic similarity of the generators for our whole class of models,
that includes the SEP(1) ,
leads us to expect multilinear correlation functions. This turns out to be false
in general, as we will see in this section.
For instance bilinearity of the covariances holds only for a certain choice
of the boundary parameters for the SEP($2j$) and for the SIP($2k$) and
only in the case $k=1/4$ for the BEP.
From Proposition \ref{corollario-abs}, multilinearity of the correlation functions
is in turn implied by multilinearity of the absorption probabilities of
the dual walkers.
\vskip.3cm

In this section we  compute the two points correlations w.r. to the stationary measure,  for a suitable choice of the parameters. We do this  by direct computation, i.e. for the particle models
we require that the generator of the process vanishes on the functions  $f(\eta)=\eta_i \, \eta_\ell$.
This yields a  linear system in the variables $X_{i,\ell}:=\langle \eta_i \eta_\ell \rangle$, $i\le \ell$. Analogous computations, with $\eta_i$ replaced
 by $z_i$, are performed for the BEP model.

\vskip.5cm

\subsection{Covariances} \label{COV}

The correlations in the stationary state, i.e. the expectations $X_{i,\ell} = \langle \eta_i\eta_\ell \rangle$
with $1\le i \le \ell \le L$,
satisfy a system of $L(L+1)/2$ equations.
The equations  are quite complicated   (we include them  in the Appendix, see section \eqref{SECT:Eq}) and then hard to solve  directly.
What we found is that, for a generic choice of the boundary parameters, for none of our processes there  exists a bilinear function satisfying them.
In other words, the ansatz
\begin{eqnarray}\label{BilAnsatz}
X_{i,\ell}=A i \ell +B i+ C \ell + D \quad\quad  \text{for} \quad i<\ell \quad \quad \quad \quad
\text{and}   \quad \quad X_{i,i}=E i^2 + F i + G
\end{eqnarray}
 does not produce a solution for  the  systems in  \eqref{EqCSipSep} and \eqref{EqCBep}
 (since the number of independent equations that the coefficients in \eqref{BilAnsatz} should satisfy is larger than 7).
But there exist some conditions on  $\a, \b, \ga, \d$  producing an effective simplification of the systems \eqref{EqCSipSep}.  Under this conditions the correlations for the SIP($2k$) and for the SEP($2j$) are bilinear  and then explicitly computable through the ansatz \eqref{BilAnsatz}. In what follows we provide such explicit bilinear forms. In order to verify their validity one can simply put the generic  forms \eqref{BilAnsatz} in the systems  \eqref{EqCSipSep}, find the equations that must be satisfied by the 7 coefficients and solve them.

Finally, at the end of the paragraph,   we will see by duality that one needs to fix $k=1/4$ in order to have bilinear correlations for the BEP.

\vskip.3cm

We denote by  $\langle \eta_i \eta_\ell\rangle_c$ the covariances (truncated correlations) in the
stationary state of the particle models, namely $\langle\eta_i \eta_\ell\rangle_c:= \langle\eta_i \eta_\ell\rangle- \langle \eta_i \rangle  \langle\eta_\ell \rangle$.
Replacing $\eta_i$ with $z_i$ one defines the covariances of the BEP model. \\ \\

{\bf Inclusion walkers SIP($2k$).}
If the parameters satisfy the condition
\begin{equation}
\ga= 2k+\a \quad \quad \text{and} \quad \quad \b= 2k+\d\;,
\end{equation}
i.e.
\begin{equation}
\rho_a=2k \, \frac {\a}{\ga-\a}= \a  \quad \quad \text{and} \quad \quad \rho_b=2k \, \frac {\d}{\d-\b}= \d\;,
\end{equation}

one has
\begin{eqnarray}\label{sip-cov}
\langle \eta_i \, \eta_\ell\rangle_c & = &\frac{i  (L+1-\ell)}{(L+1)^2 (2 k (L+1)+1)}\, (\rho_a -\rho_b )^2 \quad \quad \quad \text{for} \quad \quad i<\ell\,,
\end{eqnarray}
whereas $\langle \eta_i^2 \rangle$ is a quadratic function of $i$.
Notice that, under this same choice of parameters, the expression for the averages is simplified to
\begin{eqnarray}
\label{g3}
\langle \eta_i \rangle = \rho_a + (\rho_b-\rho_a)\frac{i}{L+1}  \;.
\end{eqnarray}

{\bf Exclusion walkers SEP($2j$).}
Under the choice of the parameters
\begin{equation}
\ga= 2j-\a \quad \quad \text{and} \quad \quad \b= 2j-\d\;,
\end{equation}
i.e.
\begin{equation}
\rho_a=2j \, \frac {\a}{\a+\ga}= \a  \quad \quad \text{and} \quad \quad \rho_b=2j \, \frac {\d}{\b+ \d}= \d\;,
\end{equation}

the two points correlations are bilinear and they are given by
\begin{eqnarray}\label{sep-cov}
\langle\eta_i \, \eta_\ell\rangle_c & = & - \frac{ i \, (L+1-\ell)}{(L+1)^2 (2 j (L+1)-1)} \,  (\rho_a -\rho_b )^2  \quad \quad \quad \text{for} \quad \quad i<\ell\;,
\end{eqnarray}
the variances are quadratic and
the average profile becomes
\begin{eqnarray}
\label{g1}
\langle \eta_i \rangle = \rho_a + (\rho_b-\rho_a)\frac{i}{L+1}  \;.
\end{eqnarray}



{\bf Brownian energy process BEP($1/2$).}
In \cite{GKR} it was studied
the BEP model for $k=1/4$ and it was found that
the two points correlations are bilinear. For $i<\ell$ they are given by:
\begin{eqnarray}
\label{bep-cov}
\langle z_i \, z_\ell\rangle_c & = & \frac{2i(L+1-\ell)}{(L+3)(L+1)^2}\, (\theta_b-\theta_a)^2\;.
\end{eqnarray}

In this case
 one has the neat linear profile
\begin{eqnarray}
\label{g2}
\langle z_i \rangle = \theta_a + (\theta_b-\theta_a)\frac{i}{L+1} \;.
\end{eqnarray}

The result in \eqref{bep-cov} can be obtained  from \eqref{sip-cov} and duality. Indeed, comparing (\ref{dual-inc}) and (\ref{dual-bep}),
one notices that the dual processes of BEP($2k$) and SIP($2k$) with $\g-\a=2k$ and $\b-\d=2k$ do coincide if
and only if $k=1/4$.
Under this choice, when the dual process is initialized from the configuration
$\bar{\xi}$ having one particle at site
$i$ and one particle at site $\ell$, equation (\ref{dualSIP}) becomes
\begin{equation}
D^{SIP}(\eta,\bar{\xi})=\(2\a \)^{\xi_0} \; 4 \eta_i \eta_\ell \; \(2 \d\)^{\xi_{L+1}}
\end{equation}
and equation (\ref{dual2}) becomes
\begin{equation}
D^{BEP}(z,\bar{\xi})=(2T_a)^{\xi_0} \; 4  z_i z_\ell \; (2T_b)^{\xi_{L+1}}\;.
\end{equation}
Therefore, with this choice of parameters and initial conditions, the duality functions
are the same if one identifies $T_a = \a = \rho_a$ and $T_b = \d = \rho_b$
and the result (\ref{bep-cov}) immediately follows from (\ref{sip-cov}).

\vskip.3cm

\vskip.3cm

We can summarize the situation as follows. The covariances are bilinear at least in the following cases:
\begin{enumerate}[(a)]
\item SEP($1$), ($j=1/2$) and generic $\alpha,\beta,\gamma,\delta$.
\item SEP($2j$)  for $j\in \{1,3/2,2,5/2,...\}$ and $\gamma+\alpha = 2j$ and $\beta+\delta = 2j$.
\item SIP($2k$)  for $k>0$ and $\gamma-\alpha = 2k$ and $\beta-\delta = 2k$.
\item BEP($\frac{1}{2}$), ($k=1/4$) and generic $T_a,T_b$.
\end{enumerate}
We remark that the conditions (b),(c),(d) are those giving the neat average profile
of equations (\ref{g3}), (\ref{g1}) and (\ref{g2}), i.e. those yielding exactly the densities $\rho_a$ and $\rho_b$ (resp. the temperatures $T_a$ and $T_b$)
in the proximity of the reservoirs (i.e. for $i=0$ and $i=L+1$).

\vskip.3cm
The following further properties of the covariances are observed by solving
the equations in the Appendix \eqref{SECT:Eq}) on {\em Mathematica}.
As the parameters are varied, the covariances are:
\begin{enumerate}[(i)]
\item proportional to $(\rho_b -\rho_a)^2$ or $(T_b-T_a)^2$.
\item {positive} for the inclusion walkers and for the Brownian energy process,
{negative} for the exclusion walkers: this is related to the attractive (bosonic)
interaction of the first two system, compared to the repulsive (fermionic) interaction
of exclusion. For the proof of this property see \cite{GRV}.
\end{enumerate}

\subsection{Results for the $n$-points correlations}

The multilinearity of the correlations seems to be, thus, prerogative of some special cases.
One may wonder about the multilinearity for the $3$-points correlations, in the same range of
parameters leading to bilinearity for the $2$-points correlations (i.e. in the cases (b),(c),(d) above).
The explicit solution of the $n$-points correlations problem becomes harder and harder
as $n$ increases and even the case $n=3$  is quite difficult to  solve exactly.

In this paragraph we provide the results of some numerical computations. We solved numerically the master equation for the invariant distribution of SEP($2j$) in the cases $L=6$ and $j=1/2, 1, 3/2, 2$ and computed the correlations $\langle \eta_i \eta_j\rangle_c$ and $\langle \eta_i \eta_j \eta_\ell\rangle_c$. If  $\langle \eta_1 \eta_{i}\rangle_c$ were multilinear, the differences $d_i=\langle \eta_1  \eta_{i+1}\rangle_c - \langle \eta_1 \eta_{i}\rangle_c,\,i=2,3,4,5$ would be constant.
The simulations seem to confirm the bilinearity of the covariances in the cases (a) and (b) above, and the loss of bilinearity  in the other cases. In Fig.1 (left panel) the values of $d_i$ are reported for the case $\a=1,\, \g=1,\, \b=1/2,\, \d=3/2$: they are clearly constant for $j=1/2$ (case (a) above) and  for $j=1$ (case (b) above) but not for $j=3/2$ and $j=2$.\\

Concerning the 3-points correlations, the simulations  show that the multilinearity is lost
even in the cases where it holds for $n=2$ (i.e. in the case (b)), while it
is conserved for the SEP(1) with at most one particle per site.
Figure 1 gives evidence for this phenomenon
by showing that  $e_i=\langle \eta_1 \eta_2 \eta_{i+1}\rangle_c - \langle \eta_1 \eta_2 \eta_{i}\rangle_c,\, i=3,4,5$ are constant only for $j=1/2$ (case (a) above).\\
The deviation from multilinearity is in any case very small and, very likely, it is decreasing  as $L$ increases.
\begin{figure}[h]
\centering
\vspace{2cm}
\includegraphics[width=7.5cm]{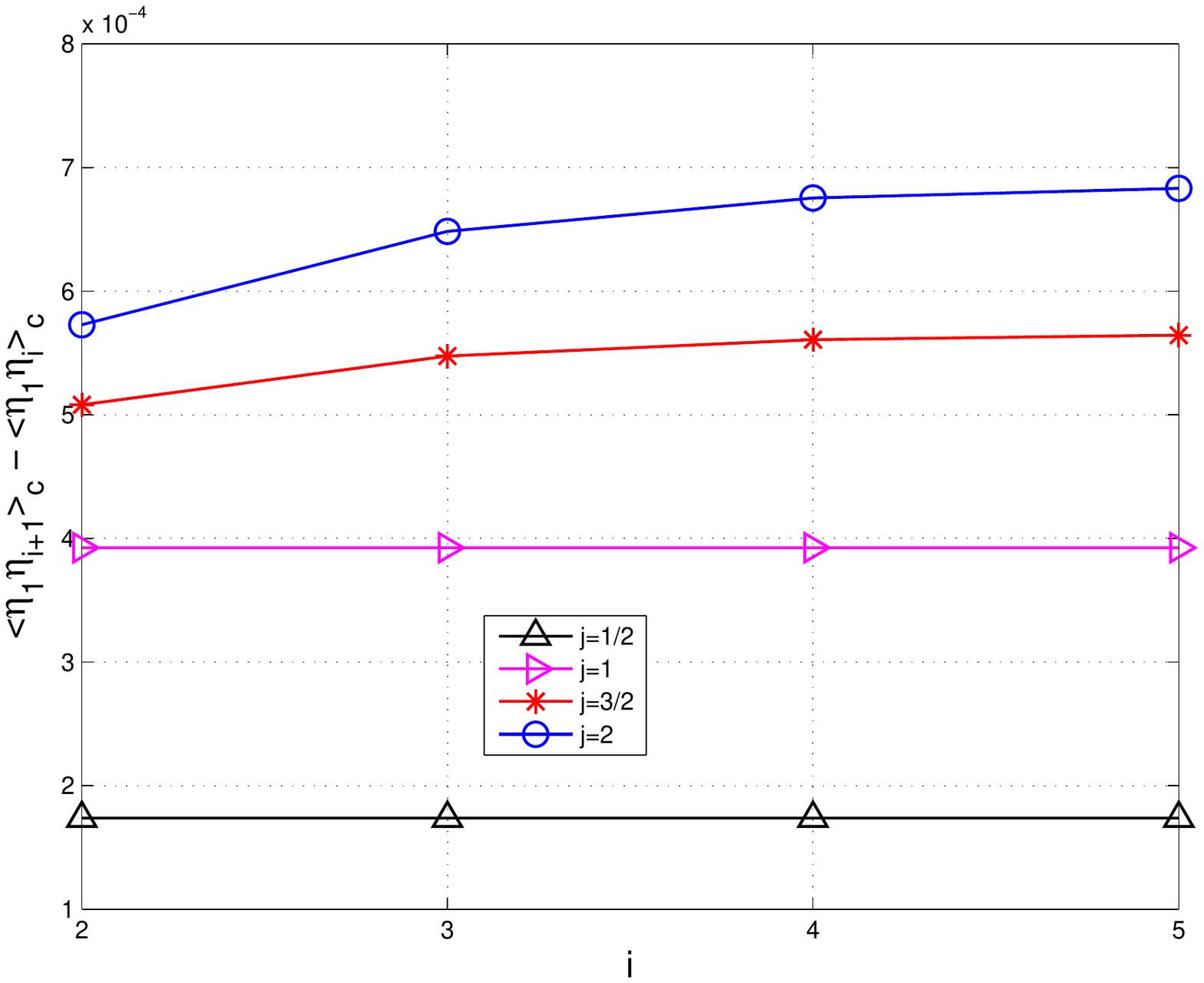}
\includegraphics[width=7.5cm]{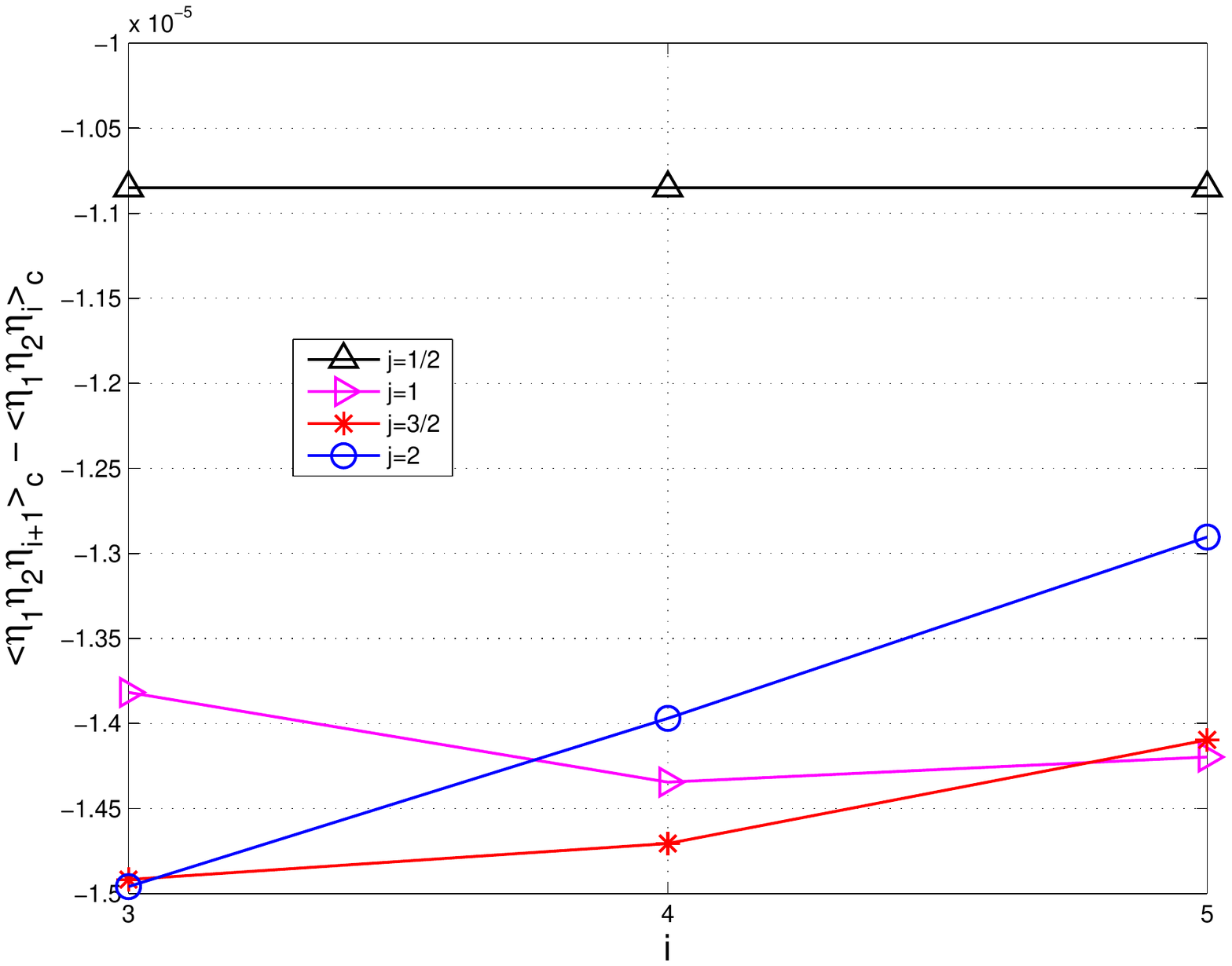}
\caption{\label{fig1}  Test for the multilinearity of the connected correlations  $\langle \eta_1 \eta_{i}\rangle_c$, $\langle \eta_1 \eta_2 \eta_{i}\rangle_c$  for SEP($2j$), $j=1/2\, (\triangle)$, $j=1\, (\rhd)$, $j=3/2\, (\ast)$, $j=2\, (\circ)$ with $\a=1,\, \g=1,\, \b=1/2,\, \d=3/2$ and $L=6$. In the left panel $d_i=\langle \eta_1 \eta_{i+1}\rangle_c - \langle \eta_1 \eta_{i}\rangle_c,\, i=2,3,4,5$; in the right panel $e_i=\langle \eta_1 \eta_2 \eta_{i+1}\rangle_c - \langle \eta_1 \eta_2 \eta_{i}\rangle_c,\, i=3,4,5$. Non constant $d_i$ or $e_i$ imply
violation of the multilinearity.}
\end{figure}

\section{Macroscopic fluctuation theory}
\label{hydro}
The aim of this section is to show that the large scale properties of the models studied so far
can be obtained by a suitable adaptation of the macroscopic fluctuation theory of  \cite{BDGJL,BDGJL1,BDGJL2, BDGJL3}.
In particular we verify that the macroscopic limit of the exact solutions for the covariances  found in Section \ref{compute} does match the prediction
of the macroscopic fluctuation theory (see \cite{DLS,DG} for the exclusion process with at most one particle per site).

\subsection{Macroscopic fluctuation theory and density large deviations functional}
We briefly review the approach of the macroscopic fluctuation theory. Let us consider a one dimensional diffusive systems of linear size $L$ in contact with two reservoirs at densities $\rho_a$ and $\rho_b$.
The macroscopic fluctuation theory describes the behavior of the system in the hydrodynamic limit in terms of the two quantities
 $D(\rho)$ and $\sigma(\rho)$, called {\em diffusivity}  and {\em mobility}, defined by
\bea
D(\rho)&:= &\lim_{\delta \rho \to 0} \; \lim_{L \to \infty} \;\frac{L}{\delta \rho} \,\cdot \, \frac {\langle Q_{i,i+1}(t)\rangle_{L, \rho, \rho + \delta \rho}}{t}\;, \label{Diff}\\
\sigma(\rho)&:=& \lim_{L \to \infty} \frac{\langle Q^2_{i,i+1}(t)\rangle_{L, \rho, \rho}}{t}\;, \label{Mob}
\ea
where
\be
Q_{i,i+1}(t) = \int_0^t j_{i,i+1}(t')dt'\;.
\ee
In the above equation $Q_{i,i+1}(t)$ is the total flow through the bond $i,i+1$ in the time interval $[0,t]$, while  $j_{i,i+1}(t')$ is the instantaneous flow at time $t'$.
The bracket $\langle \cdot \rangle_{L, \rho_a, \rho_b}$ denotes the expectation with respect to the stationary state for the system of size $L$ whose density on the left (resp. right) boundary is $\rho_a$ (resp. $\rho_b$).

From the macroscopic fluctuation theory \cite{BDGJL2}, we know that the probability of observing a time dependent
density and current profiles $\rho(x,\tau)$ and $j(x,\tau)$ in a macroscopic time interval $[\tau_1,\tau_2]$, under the diffusive
scaling $x=i/L$ and $\tau= t/L^2$, is
$\sim exp(-{L \cal A})$, where $\cal A$ is the action functional given by:
\be
{\cal A}(\{\rho(x,s),j(x,s)\};\tau_1,\tau_2)=  \int_{\tau_1}^{\tau_2}ds \int_0^1 dx \frac{\left [ j(x,s)+D(\rho(x,s)) \frac{\displaystyle \partial \rho(x,s)}{\displaystyle \partial x}\right ]^2}{2\sigma(\rho(x,s))}\;.
\ee
Then the probability of observing a density profile $\rho(x)$ in the stationary state is
$\mathbb{P}(\rho(x)) \sim e^{-{L \cal F}(\rho(x))}$ where $\cal F$ is the large deviation functional:
\be \label{ldF}
{\cal F}(\rho(x)) = \min_{\{\rho(x,s),j(x,s)\}}{\cal{A}}(\{\rho(x,s),j(x,s)\};-\infty,\tau)
\ee
with the minimum in \eqref{ldF}  taken over all the trajectories conditioned to the extreme values $\rho(x, -\infty) = \rho^*(x)$,  $\rho^*(x)$ the typical
profile, and $\rho(x,\tau)=\rho(x)$. Density and current must also satisfy the continuity equation
\be \label{CL}
\partial \rho / \partial \tau =- \,\partial j / \partial x.
\ee
\vskip.3cm
The  density correlation functions in the stationary state can be obtained from the large deviation functional $\cal F$ through
 the derivatives of its Legendre transform $\cal G$  (see \cite{D}):
\begin{equation}\label{G2}
\mathcal G (\{\a(x)\})=\sup_{\{\rho(x)\}}\left\{ \int_0^1 \a(x) \rho(x)\, dx- \mathcal F (\{\rho(x)\})\right\}.
\end{equation}
 Then, for large $L$ we have
\bea \label{GSEP}
\langle \rho(x)\rangle &=& \frac{\partial {\mathcal G}}{\partial \a(x)} \bigg |_{\a(x)=0} \\
\langle \rho(x)\rho(y)\rangle_c &=&\frac{1} L \, \frac{\partial^2 {\mathcal G}}{\partial \a(x) \, \partial \a(y)} \bigg |_{\a(x)=0} \nonumber\\
\vdots && \vdots\\
\langle \rho(x_1)\rho(x_2) \;\dots\,\rho(x_k) \rangle_c  &=&\frac{1} {L^{k-1}} \, \frac{\partial^k {\mathcal G}}{\partial \a(x_1) \, \dots \; \partial \a(x_k)} \bigg |_{\a(x)=0}\;.\nonumber
\ea
\subsection{From SEP(1) to models with constant diffusivity and quadratic mobility}
In this section we use a scaling argument to deduce the density large deviations functional
of a model with constant diffusivity and quadratic mobility from that of the SEP(1) (cfr. also \cite{BDGJL2}).
We start by recalling that for the SEP(1) one has
\be\label{diffmobsep1}
D(\rho)=1, \quad \sigma(\rho)=2\rho(1-\rho)\;,
\ee
and therefore
\be
{\cal A}_{SEP(1)}(\{\rho(x,s),j(x,s)\};\tau_1,\tau_2)=  \int_{\tau_1}^{\tau_2}ds \int_0^1 dx \frac{\left [ j(x,s)+ \frac{\displaystyle \partial \rho(x,s)}{\displaystyle \partial x}\right ]^2}{4\rho(x,s)(1-\rho(x,s))}
\ee
then, from \eqref{ldF}, one finds that the density large deviation functional is (see \cite{DLS}, \cite{BDGJL3})
\be
{\cal F }_{SEP(1)}(\{\rho(x)\})=  \sup_{F(x)}\int_0^1 dx \left [ \rho(x) \log \frac{\rho(x)}{F(x)} + (1-\rho(x)) \log\left (\frac{ 1-\rho(x)}{1-F(x)}\right)+ \log \frac{F^\prime(x)}{\rho_a-\rho_b}\right ]
\ee
where the supremum is taken over the monotone functions with boundary values $F(0)=\rho_a$, $F(1)=\rho_b$.
The supremum is attained for  $F=F_\rho$, monotone  solution of the following differential problem:
\begin{equation}\label{F}
\rho(x)= F+ \frac{F(1-F)F''}{(F')^2} \quad \quad \quad \text{with} \quad\quad \quad F(0)=\rho_a \quad \text{and} \quad F(1)=\rho_b.
\end{equation}

The connected correlation functions can be obtained by computing the derivatives of
the functional $\mathcal G_{SEP(1)} $ as in \eqref{G2}.
One finds that  the  lowest order correlations are, for large $L$,
\bea \label{corrSEP}
\langle \rho(x)\rangle^{SEP(1)} &=& \rho_a(1-x)+\rho_b x  \\
\langle \rho(x)\rho(y)\rangle_c^{SEP(1)} &=& - \frac{(\rho_a - \rho_b)^2}{L}x(1-y)\nonumber\\
\langle \rho(x)\rho(y)\rho(z) \rangle_c^{SEP(1)}  &=& -2 \frac{(\rho_a - \rho_b)^3}{L^2} x (1-2y) (1-z), \nonumber
\ea
for $0<x<y<z$.

\vskip.5cm

Let us now  consider the generalization of (\ref{diffmobsep1})  obtained by assuming that the diffusivity is constant and the mobility is a    quadratic function parametrized as
\be\label{diffmobgen}
D(\rho)=C,\quad \sigma(\rho)= 2 A \rho (B-\rho),
\ee
where $A$, $B$ and $C$ are given numbers. The action functional of this system
 \be
{\cal A}(\{\rho(x, s), j(x, s)\}; \tau_1, \tau_2)=  \int_{\tau_1}^{\tau_2} ds \int_0^1 dx \frac{\left [ j(x,s)+C \frac{\displaystyle \partial \rho(x,s)}{\displaystyle \partial x}\right ]^2}{4 A\rho(x,s)(B-\rho(x,s))}
\ee
is related to ${\cal A}_{SEP(1)}$ through the following change of variables (cfr. \cite{DG})
\be \label{AA}
{\cal A}(\{\rho(x,\tau), j(x,\tau)\}; \tau_1, \tau_2)= \frac{C}{A} \,{\cal A}_{SEP(1)}(\{\tilde \rho(x, s), \tilde j(x, s)\};C \tau_1, C \tau_2)\;,
\ee
with
\be
 \tilde \rho(x, s):= \frac 1 B \,\rho(x,C^{-1} s)\quad \quad \text{and} \quad\quad  \tilde j(x, s):=\frac 1 {BC} \,j(x,C^{-1}s)
\ee
The scaling \eqref{AA} has been chosen among all the possible scalings connecting ${\cal A}$ and ${\cal A}_{SEP(1)}$ as it is the only one
 preserving  the  conservation law \eqref{CL} between $\tilde \rho(x, s)$ and $\tilde j(x, s)$.
\\
Then, by  \eqref{AA} and  \eqref{ldF} it follows that
 the large deviation functional for the system characterized by  (\ref{diffmobgen}) is given by
 \be \label{FF}
 {\cal F }(\{\rho(x)\})=\frac{C}{A}\, {\cal F }_{SEP(1)}\(\{B^{-1}\rho(x)\}\) \;,
 \ee
and thus
\bea\label{FreeEn}
{\cal F }(\{\rho(x)\})= \frac{C}{AB} \,\sup_{\tilde F(x)} \int_0^1 dx \left [ \rho(x) \log \frac{\rho(x)}{\tilde F(x)} + (B-\rho(x)) \log \left (\frac{ B-\rho(x)}{B-\tilde F(x)}\right )+B \log \frac{\tilde F^\prime(x)}{\tilde\rho_a-\tilde\rho_b}\right ]\;,\nonumber\\
\ea
with
\be
 \tilde \rho_a= B \rho_a,  \quad\quad \tilde \rho_b= B \rho_b \quad \quad \quad \text{and} \quad \quad \quad \tilde F(x)=B F(x) \;,
\ee
where $F(x)$ is the monotone function satisfying \eqref{F}.
Equivalently $\tilde F$ is the monotone solution of the differential problem
\be \label{tF}
\rho(x)= \tilde F+ \frac{\tilde F(B-\tilde F)\tilde F''}{(\tilde F')^2} \quad \quad \quad \text{with} \quad\quad \quad \tilde F(0)=\tilde \rho_a \quad \quad  \text{and} \quad \quad \tilde F(1)=\tilde \rho_b.
\ee

Using \eqref{G2}
and \eqref{FF} we find
\bea
{\cal G}\(\{\a(x)\}\)&=&\sup_{\{\tilde \rho(x)\}}\left\{ \int_0^1 \a(x) \tilde \rho(x)\, dx- \mathcal F(\{\tilde \rho(x)\})\right\}\nonumber\\
&=& \sup_{\{\tilde \rho(x)\}}\left\{ \int_0^1 \a(x) \tilde \rho(x)\, dx-\frac {C}{A} \,\mathcal F_{SEP(1)}(\{B^{-1}\tilde\rho(x)\})\right\}\nonumber\\
&=&\frac{C}{A} \, \sup_{\{\rho(x)\}}\left\{ \int_0^1 \frac {A B} C \,\a(x)  \rho(x)\, dx- \mathcal F_{SEP(1)}(\{\rho(x)\})\right\}\nonumber\\
&=& \frac{C}{A} \,{\cal G}_{SEP(1)}\(\{A B C^{-1}\a(x)\}\)
\ea

and, from \eqref{GSEP} and \eqref{corrSEP}, we have
\bea \label{corr}
\langle \rho(x)\rangle  &=& \frac{\partial {\mathcal G}}{\partial \a(x)} \bigg |_{\a(x)=0} = B\, \langle \rho(x)\rangle^{\small{SEP(1)}}=  \tilde\rho_a(1-x)+\tilde\rho_b x\\
\nonumber\\
\langle \rho(x)\rho(y)\rangle_c &=& \frac{1} L \, \frac{\partial^2 {\mathcal G}}{\partial \a(x) \, \partial \a(y)} \bigg |_{\a(x)=0} =
 \frac{A B^2}{C}\, \langle \rho(x)\rho(y)\rangle_c^{SEP(1)} =  - \left (\frac{A}{C}\right )\frac{(\tilde\rho_a - \tilde\rho_b)^2}{L}x(1-y)\nonumber\\
 \nonumber\\
\langle \rho(x)\rho(y)\rho(z) \rangle_c  & =& \frac{1} {L^2} \, \frac{\partial^3 {\mathcal G}}{\partial \a(x) \, \partial \a(y) \partial \a(z)} \bigg |_{\a(x)=0}= \frac{A^2 B^3}{C^2}\,\langle \rho(x)\rho(y)\rho(z) \rangle_c^{SEP(1)} \nonumber\\
\nonumber\\
 &&\hskip4.5cm = -2\left (\frac{A}{C}\right )^2 \frac{(\tilde\rho_a - \tilde\rho_b)^3}{L^2} x (1-2y) (1-z) \nonumber
\ea
and, more generally, one gets a factor $B^n (A/C)^{n-1}$ for the $n$-point connected correlation function.\\

\vskip1cm
\subsection{Macroscopic behavior of the correlations}
With suitable choices of the parameters $A,B,C$ we can generate the large scale limits of models that we have considered in the previous sections.

\vskip.5cm
{\bf Inclusion walkers SIP($2k$)}.
For interacting particle systems, the flux across bond $(i,i+1)$ in a time interval $[0,t]$
is given by the number of particles which jump from $i$ to $i+1$ minus the number of particles
which jump from $i+1$ to $i$. i.e.
\begin{eqnarray}
\label{flusso}
Q_{i,i+1}(t)
& = &
\int_0^t dt' \left[\eta_{i+1}(t')(2k+\eta_i(t'))-\eta_i(t')(2k+\eta_{i+1}(t'))\right] \nonumber \\
& = &
2k \int_0^t dt' \left[\eta_{i+1}(t')-\eta_i(t')\right] \;.
\end{eqnarray}
As a consequence, the expectation of the flow $Q_{i,i+1}(t)$ in the stationary state
with boundary densities $\rho_a, \rho_b$ is given by
\be
\langle Q_{i,i+1}(t)\rangle_{L,\rho_a, \rho_b}= 2k \cdot t \cdot
\langle \eta_{i+1} -  \eta_{i}\rangle_{L,\rho_a, \rho_b}\;.
\ee
It follows, from \eqref{Diff} and \eqref{Exp},  that  $D(\rho)=2k$.

From Section \ref{Canonical} we know that the SIP($2k$) equilibrium stationary measure
at density $\rho$ is the product of {\tt NegativeBinomial} $(2k, \rho/(\rho +2k))$ with a variance
$Var(\eta_i)=\frac{\rho(\rho+2k)}{2k}$. Using (\ref{flusso})
\begin{eqnarray}
 \langle Q^2_{i,i+1}(t)\rangle_{L,\rho, \rho}
& = &
(2k)^2 \int_0^t dt' \int_0^t dt'' \langle \left[\eta_{i+1}(t')-\eta_i(t')\right] \left[\eta_{i+1}(t'')-\eta_i(t'')\right] \rangle_{L,\rho,\rho} \;.
\end{eqnarray}
Now we have
\begin{eqnarray}
&&\lim_{t \to \infty}\frac 1t \int_0^t dt' \int_0^t dt'' \langle(\eta_{i+1}(t')-\eta_{i}(t'))(\eta_{i+1}(t'')-\eta_{i}(t'')) \rangle_{L, \rho, \rho} \nonumber
\\
&=& \lim_{t \to\infty}\frac 2{t} \int_0^t dt' \int_{t'}^t dt'' \langle(\eta_{i+1}(t')-\eta_{i}(t'))(\eta_{i+1}(t'')-\eta_{i}(t'')) \rangle_{L, \rho, \rho} \nonumber\\
&=&2\int_0^\infty dt \,\langle(\eta_{i+1}(0)-\eta_{i}(0))(\eta_{i+1}(t)-\eta_{i}(t)) \rangle_{L,\rho,\rho}  \nonumber\\
&=& 2 \int_0^\infty dt \cdot \sum_{\eta} \left\{  (\eta_{i+1}-\eta_{i}) \,\mathbb E_{\eta} \[\eta_{i+1}(t)-\eta_{i}(t)\] \; \mu_{L,\rho,\rho}(\eta) \right\}\label{su}
\end{eqnarray}
where,  in the last display, $\mu_{L,\rho,\rho}$ denotes the stationary equilibrium measure.
By duality,
\begin{eqnarray}
\mathbb E_{\eta} \[\eta_{i+1}(t)-\eta_{i}(t)\]
&=&{2k}\left\{ \mathbb E_{\eta} \[D_0^{SIP}(\eta(t), \xi^{i+1})\]-\mathbb E_{\eta} \[D_0^{SIP}(\eta(t), \xi^i)\] \right\}\nonumber\\
&=&{2k}\left\{ \mathbb E_{\xi^{i+1}} \[D_0^{SIP}(\eta, \xi(t))\]-\mathbb E_{\xi^i} \[D_0^{SIP}(\eta, \xi(t))\] \right\}
\end{eqnarray}
where $\xi^i$ is the $L$-dimensional configuration $(\xi_1^i, \dots, \xi_{L}^i)$ with $\xi^i_j= \delta_{i,j}$ and $D_0^{SIP}$ is the duality function defined in \eqref{D0}.
Let $p_t(i,j)$ be the transition probability from the site $i$ to the site $j$ in the time interval $[0,t]$ of a  random walker on the set $\{1, \dots, L\}$ moving at rate $2k$, then
\be
\mathbb E_{\xi^i} \[D_0^{SIP}(\eta, \xi(t))\]  =\frac 1 {2k}\sum_j  \eta_j \cdot p_t(i,j)\;.
\ee
As a consequence \eqref{su} is equal to
\bea \label{qua}
&&2\sum_{j=1}^L  \langle(\eta_{i+1}-\eta_{i})\eta_j \rangle_{L,\rho,\rho} \cdot \int_0^\infty dt \, \(p_t (i+1,j)- p_t(i,j)\)  \nonumber\\
&&= 4 \, \langle(\eta_{i+1}-\eta_{i})\eta_{i+1} \rangle_{L,\rho,\rho}  \cdot \int_0^\infty dt \, \(p_t (i+1,i+1)- p_t(i,i+1)\) \nonumber \\
&&= 4 \, Var (\eta_i) \cdot \int_0^\infty dt \, \(p_t (0,0)- p_t(0,1)\)
\ea
where the two identities above follow from the product character of the equilibrium measure, and from the fact that $p_t(i,j)$ depends only on the distance $|i-j|$.
Now the random walk $p_t$ is moving at rate $2k$, then, from the master equation we have
\be
2 \(p_t (0,0)- p_t(0,1)\) = -\(p_t(0,1)+p_t(0,-1)-2p_t(0,0)\)
=-\, \frac 1 {2k} \cdot \frac {d}{dt} p_t(0,0)\;.
\ee
Then \eqref{qua} is given by
\be
 - 2 Var (\eta_i) \cdot \int_0^\infty \frac{1}{2k} \frac d {dt} \, p_t(0,0) \cdot dt= 2 \,\frac 1 {(2k)^2} \cdot \rho(\rho+2k)\cdot \(1-p_{\infty}(0,0)\)\;.
\ee
Since $p_\infty(0,0)$ vanishes as $L \to \infty$, we finally  obtain, using (\ref{Mob}) $\sigma(\rho)=2\rho (\rho +2k)$.

Summarizing, for the inclusion process SIP($2k$), we have
\be \label{Dsipa}
D(\rho)=2 k, \quad \sigma(\rho)=2\rho (\rho +2k)\;,
\ee
which implies  $A=-1$, $B=-2k$ and $C=2k$. This choice produces (see \eqref{corr})  the following correlation functions:
\bea \label{corrSIP}
\langle \rho(x)\rangle &=& \rho_a(1-x)+\rho_b x \nonumber \\
\langle \rho(x)\rho(y)\rangle_c &=&\frac{1}{2k}\frac{(\rho_a - \rho_b)^2}{L}x(1-y) \nonumber \\
\langle \rho(x)\rho(y)\rho(z) \rangle_c  &=& -\left (\frac{1}{2k} \right )^2\frac{2(\rho_a - \rho_b)^3}{L^2} x (1-2y) (1-z).
\ea
where $\rho_a$ and $\rho_b$ are the SIP($2k$) boundary densities ($\rho_a= 2k \, \a /(\ga-\a)$ and $\rho_b=2 k \, \d /(\b-\d)$ in terms of our boundary parameters).
Notice that the covariances in \eqref{corrSIP} do indeed agree with the macroscopic limit of the
microscopic covariances that have been found in Section \ref{COV} (see \eqref{sip-cov}) for a
particular choice of the parameters.
Similarly, one gets for  the density large deviation functional:
\be
{\cal F }(\{\rho(x)\})= \int_0^1 dx \left [ \rho(x) \log \frac{\rho(x)}{F(x)} + (2k+\rho(x)) \log \left (\frac{ 2k+\rho(x)}{2k+F(x)}\right )+2k \log \frac{F^\prime(x)}{\rho_a-\rho_b}\right ]\;,
\ee
where  $F=F_\rho$ is the monotone solution of
\be
 \rho(x)= F + \frac{ F(2k+ F) F''}{( F')^2} \quad \quad \quad \text{with} \quad\quad \quad  F(0)= \rho_a \quad \quad  \text{and} \quad \quad  F(1)= \rho_b.
\ee

\vskip.5cm

{\bf Exclusion walkers  SEP($2j$).}
The flux is now given by
\begin{eqnarray}
Q_{i,i+1}(t)
& = &
\int_0^t dt' \left[\eta_{i+1}(t')(2j-\eta_i(t'))-\eta_i(t')(2j-\eta_{i+1}(t'))\right] \\
& = &
2j \int_0^t dt' \left[\eta_{i+1}(t')-\eta_i(t')\right] \;.
\end{eqnarray}
As a consequence, the expectation of  $Q_{i,i+1}(t)$   with respect to the steady state measure
reads
\be
\langle Q_{i,i+1}(t)\rangle_{L,\rho_a, \rho_b}= 2j \cdot t \cdot \langle \eta_{i+1}-  \eta_{i}\rangle_{L,\rho_a, \rho_b}\;.
\ee
Thus, from \eqref{Diff} and \eqref{Exp},  we get $D(\rho)=2j$.
From Section \ref{Canonical} we know that the SEP($2j$) stationary measure at density $\rho$ is the product of {\tt Binomial} $(2j, \rho/2j)$
with a variance $Var(\eta_i)=\frac{\rho(2j-\rho)}{2j}$. Using a similar computation as for the inclusion walkers then, one can compute also the mobility,
obtaining:
\be \label{Dsep}
D(\rho)=2 j, \quad \sigma(\rho)=2\rho (2j-\rho).
\ee

Therefore we have  $A=1$, $B=2j$, $C=2j$ and, from \eqref{corr}, we have the following correlation functions:
\bea \label{corrSEP(2j)}
\langle \rho(x)\rangle &=& \rho_a(1-x)+\rho_b x \nonumber\\
\langle \rho(x)\rho(y)\rangle_c &=& - \frac{1}{2j}\frac{(\rho_a - \rho_b)^2}{L}x(1-y) \nonumber \\
\langle \rho(x)\rho(y)\rho(z) \rangle_c  &=& -\left (\frac{1}{2j} \right )^2\frac{2(\rho_a - \rho_b)^3}{L^2} x (1-2y) (1-z)
\ea
where $\rho_a$ and $\rho_b$ are the SEP($2j$) boundary densities ($\rho_a= 2j \, \a /(\a+\ga)$ and $\rho_b=2 j \, \d /(\b+\d)$ in terms of our boundary parameters).
The second line in \eqref{corrSEP(2j)} does agree with the microscopic SEP-covariances  that have been found in Section \ref{COV} (see \eqref{sep-cov}) for a particular choice of the parameters.
Moreover the density large deviation functional is given by
\be
{\cal F }(\{\rho(x)\})= \int_0^1 dx \left [ \rho(x) \log \frac{\rho(x)}{F(x)} + (2j-\rho(x)) \log \left (\frac{ 2j-\rho(x)}{2j-F(x)}\right )+2j \log \frac{F^\prime(x)}{\rho_a-\rho_b}\right ]
\ee
where $F=F_\rho$ is the monotone function satisfying
\be
\rho(x) = F + \frac{ F(2j- F) F''}{( F')^2} \quad \quad \quad \text{with} \quad\quad \quad  F(0)= \rho_a \quad \quad  \text{and} \quad \quad  F(1)= \rho_b.
\ee

\vskip.5cm

{\bf Independent random walkers IRW.}
As  observed in \cite{DG}, the independent random walkers model, for which
\be
D(\rho)=1, \quad  \sigma(\rho)=2\rho
\ee
 is obtained in the limit as $A=B^{-1}\rightarrow 0$ and $C=1$. Under this choice, see (\ref{corr}), all the correlation functions vanish (this obviously reflects the fact that the stationary measure has a product structure, see Proposition \ref{IRWStSt}).
As $B\rightarrow \infty$, one can see from \eqref{FreeEn} that,
due to the concavity of the logarithm, the derivative $F^\prime(x)$ is constant. Therefore in this limit the optimal $F(x)$ is given by
\be
F(x)=(1-x)\rho_a+x\rho_b
\ee
and one get
\be\label{ldfrw}
{\cal F}(\{ \rho(x)\}) =\int_0^1 dx \left [ \rho(x) \log \left (\frac{\rho(x)}{(1-x)\rho_a+x\rho_b}\right) -\rho(x) + (1-x)\rho_a+x\rho_b \right ].
\ee

\vskip.3cm
{\bf KMP model.}
The expectation of  $Q_{i,i+1}(t)$   with respect to the steady state measure $\mu_{L, T_a, T_b}$
 is now given by
\bea
\langle Q_{i,i+1}(t)\rangle_{L, T_a, T_b}&=& t \cdot  \int_0^1 dx \,\langle\[x(z_i+z_{i+1})-z_i\]-\[(1-x)(z_i+z_{i+1})-z_{i+1}\] \rangle_{L,T_a, T_b}\nonumber \\&= &t \cdot \langle z_{i+1}-  z_{i}\rangle_{L,T_a, T_b}
\ea
then, from \eqref{Diff} and \eqref{ExpKMP},  we get $D(\rho)=1$.
We know that the KMP stationary measure at temperature $T$ is the product of {\tt Exponential}$(1/T)$. By a duality argument
we compute also the mobility and get
\be
D(\rho)= 1 \quad \sigma(\rho)= 2\rho^2.
\ee
The KMP model can be, then, obtained (see \cite{DG}) by taking the {\em unphysical} limit $B\rightarrow 0$, $A\rightarrow -1$ with $C=1$.
In this limit the first three connected correlations functions (see \ref{corr})  are
\bea
\langle \rho(x)\rangle &=& \rho_a(1-x)+\rho_b x \nonumber\\
\langle \rho(x)\rho(y)\rangle_c &=& \frac{(\rho_a - \rho_b)^2}{L}x(1-y) \nonumber\\
\langle \rho(x)\rho(y)\rho(z) \rangle_c  &=& -2\frac{(\rho_a - \rho_b)^3}{L^2} x (1-2y) (1-z),
\ea
which agree with (2.38) of \cite{BDGJL}. Moreover the density large deviation functional that we obtain
\be
{\cal F }(\{\rho(x)\})= -\sup_{F(x)} \int_0^1 dx \left [ 1- \frac{\rho(x)}{F(x)} +  \log \frac{\rho(x)}{F(x)}+\log \frac{F^\prime(x)}{\rho_a-\rho_b}\right ]
\ee
agrees with the same function  computed in \cite{BGL}.

\vskip.5cm

{\bf Acknowledgments.} We are extremely grateful to Bernard Derrida, with whom we
discussed some of the topics in this work. In particular we own to him the results of
Section 7 for the comparison to macroscopic fluctuation theory.

We acknowledge financial support from the by the Italian Research Funding
Agency (MIUR) through FIRB project ``Stochastic processes in interacting particle
systems: duality, metastability and their applications'', grant n. RBFR10N90W and
the Fondazione Cassa di Risparmio Modena through the International Research 2010 project.

\section{Appendix: Equations for  the two points correlations}\label{SECT:Eq}

We provide the linear systems that must be satisfied by
the two points correlation functions in the steady state, i.e. $X_{i,\ell} = \langle \eta_i\eta_\ell \rangle$
with $1\le i \le \ell \le L$.
In the following, equations 1),2),3) are obtained by letting act the generator on a couple
of sites at distance larger or equal than two,
equations 4),5),6) are derived from nearest-neighbouring sites,
equations 7),8),9) correspond to the diagonal,
equation 10) is obtained from the couple (1,L).\\


{\bf Inclusion/Exclusion walkers}: the equations for the inclusion walkers SIP($2k$) and for the exclusion walkers SEP($2j$)
are similar, with some relevant change of signs in the two cases; therefore we write them
together. With the convention to use upper symbol for inclusion and lower symbol for exclusion
in $\pm$ and $\mp$ and with the further convention that $h=k$ for SIP($2k$) and $h=j$ for SEP($2j$),
the equations read\\
\\
\begin{equation}\label{EqCSipSep}
\begin{array}{lll}
1)\; X_{i-1,\ell}+X_{i+1,\ell}+X_{i,\ell-1}+X_{i,\ell+1}-4 X_{i,\ell}=0 \quad\quad\quad \quad\quad \quad\quad \quad\quad \quad \quad \quad\quad\quad \text{for} \quad i+1 < \ell, \: i > 1, \: \ell < L \\
\\
2)\; 2 h (X_{2,\ell} + X_{1,\ell-1} +
    X_{1, \ell + 1}) - (6 h \mp \a+ \ga) X_{1, \ell} +
 2 h \a x_\ell=0 \quad\quad \quad\quad \quad\quad\quad \quad\text{for} \quad \ell > 2\\
 \\
3)\; 2 h (X_{i, L - 1} + X_{i + 1, L} +
     X_{i - 1, L}) - (6 h + \beta \mp \delta) X_{i, L} +
  2 h \delta x_i = 0    \quad\quad\quad\quad \quad\quad \quad\quad \quad\text{for} \quad i< L-1\\
\\
4)\; h X_{i,i} + h X_{i+1,i+1} + (\mp 1 - 4 h) X_{i,i+1} + h X_{i-1,i+1} +
 h X_{i,i+2} - h (x_i + x_{i+1})=0 \quad \quad \text{for}\quad 1<i < L-1\\
 \\
5)\;  2 h X_{1,1} + 2 h X_{2,2} - (2 (3 h \pm 1) + (\mp \alpha + \gamma)) X_{1, 2} +
  2 h X_{1, 3} - 2 h x_1 - 2 h (1 - \alpha) x_2 = 0\\
\\
6)\;  2 h X_{L,L} +
  2 h X_{L - 1,L-1} - (2 (3 h \pm 1) + (\beta \mp \delta)) X_{L- 1, L} +
  2 h X_{L - 2, L} - 2 h x_L - 2 h (1 - \delta) x_{L- 1} = 0\\
  \\
 7)\;  h (x_{i-1} + 2 x_i+ x_{i+1}) + (2 h \pm 1) X_{i-1,i} -
  4 h X_{i,i} + (2 h \pm 1) X_{i,i+1} = 0 \quad \quad\quad \quad\quad \quad\;\;\text{for} \quad 1<i < L\\
  \\
 8)\; 2 (2 h + (\mp \alpha + \gamma)) X_{1,1} -
 2 (2 h \pm 1) X_{1, 2} - (2 h (2 \alpha + 1) + \gamma \pm  \alpha) x_1 - 2 h x_2 - 2 h \alpha=0 \\
\\
9)\;  2 (2 h + (\beta \mp \delta)) X_{L,L} -
  2 (2 h \pm 1) X_{L- 1,L} - (2 h (2 \delta + 1) + \beta \pm \delta) x_L -
  2 hx_{L- 1} - 2 h \delta = 0\\
\\
10)\;  -(4h +\ga \mp \delta \mp \alpha +\beta) X_{1,L} +2h X_{2,L} +2h X_{1, L-1} +2h (\delta x_1 + \alpha x_L) =0
\quad\quad\quad\quad\quad\quad\\
\\
\end{array}
\nonumber
\end{equation}

\vskip.6cm


{\bf Brownian energy process BEP($2k$)}: the equations for the BEP($2k$) read
\begin{equation}\label{EqCBep}
\begin{array}{lll}
1) \; X_{i-1,\ell}+X_{i+1,\ell}+X_{i,\ell-1}+X_{i,\ell+1}-4 X_{i,\ell}=0 & \text{for} & i+1 < \ell, \quad i > 1, \: \ell < L\\
\\
2) \; 4k (X_{1,\ell-1}+X_{1,\ell+1}+ X_{2,\ell})-(1+12k)X_{1,\ell}+4k T_a \langle z_\ell \rangle=0 &\text{for}& \ell > 2\\
\\
3) \; 4k(X_{i-1,L}+X_{i+1,L}+X_{i,L-1})-(12k+1)X_{i,L}+4k T_b \langle z_i\rangle =0 &\text{for}& i < L-1\\
\\
4) \; 2k X_{i,i}+2k X_{i+1,i+1}-2(4k+1)X_{i,i+1}+2k (X_{i-1,i+1}+X_{i,i+2})=0 &\text{for}& 1<i  <  L-1 \\
\\
5) \; 4k(X_{1,1}+X_{2,2}) -(12k + 5) X_{1,2}+4k X_{1,3}+ 4k T_a \langle z_2 \rangle=0\\
\\
6) \; 4k(X_{L,L}+X_{L-1,L-1})-(12k+5)X_{L-1,L}+4k X_{L-2,L}+4k T_b \langle z_{L-1} \rangle =0\\
\\
7) \; (2k+1)X_{i-1,i}+(2k+1)X_{i,i+1}-4k X_{i,i}=0 &\text{for}& 1<i < L\\
\\
8) \; 2(2k+1)X_{1,2}-(4k+1)X_{1,1}+ 2 (2k+1)T_a \langle z_1 \rangle =0\\
\\
9)\; 2(2k+1)X_{L-1,L}-(4k+1)X_{L,L}+2(2k+1)T_b \langle z_L \rangle =0 \\
\\
10)\; 4kT_a \langle z_L \rangle + 4kT_b \langle z_1 \rangle -2(1+4k)X_{1,L}
+ 4k(X_{2,L}+X_{1,L-1}) =0\\
\end{array}
\nonumber
\end{equation}

\end{document}